\newif\ifprocs
\procstrue
\procsfalse %

\ifprocs
\documentclass[a4paper,USenglish,numberwithinsect]{lipics-v2016}
\usepackage{babel}
\else
\documentclass[11pt]{article}
\usepackage[margin=1.0in]{geometry}
\usepackage[utf8]{inputenc}
\usepackage[english]{babel}
\fi

\usepackage{amssymb}
\usepackage{amsfonts}
\usepackage{amsmath}
\usepackage{amsthm}
\usepackage{verbatim}
\usepackage{xcolor,graphicx}
\usepackage{xspace}

\usepackage{lineno} 

\ifprocs
\hypersetup{colorlinks,allcolors=blue}
\else
\usepackage{ifpdf}
\ifpdf    %
\usepackage[pdftex,colorlinks,linkcolor=blue,citecolor=blue,filecolor=blue,urlcolor=blue]{hyperref}
\else    %
\usepackage[hypertex,colorlinks,linkcolor=blue,citecolor=blue,filecolor=blue,urlcolor=blue]{hyperref}
\fi
\fi

\ifprocs
\renewcommand{\paragraph}{\subparagraph}
\else
\newtheorem{theorem}{Theorem}[section]
\newtheorem{lemma}[theorem]{Lemma}
\newtheorem{definition}[theorem]{Definition}
\newtheorem{corollary}[theorem]{Corollary}
\fi

\theoremstyle{plain}
\newtheorem{claim}[theorem]{Claim}

\def\compactify{\itemsep=0pt \topsep=0pt \partopsep=0pt \parsep=0pt}

\newcommand{\eqdef}{:=}

\providecommand{\card}[1]{\lvert#1\rvert}
\newcommand{\set}[1]{\left\{ #1 \right\}}

\DeclareMathOperator{\poly}{poly}
\DeclareMathOperator{\Count}{Count}
\DeclareMathOperator{\parity}{Par}
\DeclareMathOperator{\mincut}{mincut}
\DeclareMathOperator{\mincutset}{argmincut}
\DeclareMathOperator*{\cupdot}{\mathbin{\mathaccent\cdot\cup}}

\DeclareMathOperator*{\mycup}{\cup} %

\newcommand{\1}{{\rm 1\hspace*{-0.4ex}%
\rule{0.1ex}{1.52ex}\hspace*{0.2ex}}}

\newcommand{\R}{\mathbb R}

\newcommand{\I}{\mathcal I}

\newcommand{\T}{\mathcal T}
\newcommand{\G}{\mathcal G}

\newcommand{\eps}{\epsilon}

\newcommand{\Wlog}{without loss of generality\xspace}

\newcommand{\MTC}{minimum terminal cut\xspace}
\newcommand{\MTCs}{{\MTC}s\xspace}
\newcommand{\TCs}{TC-scheme\xspace}

\newcommand{\Raecke}{R\"{a}cke}

\title{Refined Vertex Sparsifiers of Planar Graphs%
\footnote{This work was partially supported by the Israel Science Foundation grants \#897/13 and \#1086/18, and by a Minerva Foundation grant.}
}

\ifprocs
\author[1]{Robert Krauthgamer}
\author[2]{Havana (Inbal) Rika}
\affil[1]{Weizmann Institute of Science, Rehovot, Israel\\
  \texttt{robert.krauthgamer@weizmann.ac.il}}
\affil[2]{Weizmann Institute of Science, Rehovot, Israel\\
  \texttt{havana.rika@weizmann.ac.il}}
\authorrunning{R.Krauthgamer and H.Rika} %

\Copyright{Robert Krauthgamer and Havana (Inbal) Rika}%

\subjclass{G.2.2 Graph Theory, G.2.1 Combinatorics, E.1 Data Structure}%
\keywords{Vertex Sparsifiers, Graph Algorithms, Planar Graphs,
Graph Distances, Succinct Data Structures.}

\else
\author{Robert Krauthgamer%
  \footnote{Email: \texttt{robert.krauthgamer@weizmann.ac.il}  }
\qquad
Havana (Inbal) Rika%
  \footnote{Email: \texttt{havana.rika@weizmann.ac.il}  }
\\
Weizmann Institute of Science
}

\fi

\begin{document}
\maketitle

\begin{abstract}

We study the following version of cut sparsification.
Given a large edge-weighted network $G$ with $k$ terminal vertices,
compress it into a smaller network $H$ with the same terminals,
such that every minimum terminal cut in $H$ approximates the corresponding one in $G$,
up to a factor $q\geq 1$ that is called the quality.
(The case $q=1$ is known also as a mimicking network).
We provide new insights about the structure of minimum terminal cuts,
leading to new results for cut sparsifiers of \emph{planar} graphs.

Our first contribution identifies a subset of the minimum terminal cuts,
which we call elementary, that generates all the others.
Consequently, $H$ is a cut sparsifier if and only if
it preserves all the elementary terminal cuts (up to this factor $q$).

Our second and main contribution is to refine the known bounds
in terms of $\gamma=\gamma(G)$, which is defined as the minimum number
of faces that are incident to all the terminals in a planar graph $G$.
We prove that the number of elementary terminal cuts is $O((2k/\gamma)^{2\gamma})$ (compared to $O(2^k)$ terminal cuts),
and furthermore obtain a mimicking network of size $O(\gamma 2^{2\gamma} k^4)$,
which is near-optimal as a function of $\gamma$.

Our third contribution is a duality between cut sparsification
and distance sparsification for certain planar graphs,
when the sparsifier $H$ is required to be a minor of $G$.
This duality connects problems that were previously studied separately,
implying new results, new proofs of known results,
and equivalences between open gaps.
\end{abstract}

\section{Introduction}
A very powerful paradigm when manipulating a huge graph $G$
is to \emph{compress} it,
in the sense of transforming it into a small graph $H$
(or alternatively, into a succinct data structure) that maintains
certain features (quantities) of $G$, like distances, cuts, or flows.
The basic idea is that
once the compressed graph $H$ is computed in a preprocessing step,
further processing can be performed on $H$ instead of on $G$,
using less resources like runtime and memory,
or achieving better accuracy when the solution is approximate.
This paradigm has lead to remarkable successes,
such as faster runtimes for fundamental problems,
and the introduction of important concepts,
from spanners \cite{PU89} to cut and spectral sparsifiers \cite{BK15,ST11}.
In these examples, $H$ is a subgraph of $G$ with the \emph{same} vertex set
but \emph{sparse}, and is sometimes called an edge sparsifier.
In contrast, we aim to reduce the number of vertices in $G$,
using so-called vertex sparsifiers.

In the vertex-sparsification scenario,
$G$ has $k$ designated vertices called \emph{terminals},
and the goal is to construct a small graph $H$ that contains these terminals,
and maintains some of their features inside $G$, like distances or cuts.
Throughout, a \emph{$k$-terminal network}, denoted $G=(V,E,T,c)$,
is an undirected graph $(V,E)$ with edge weights $c:E\to \R_+$
and terminals set $T\subset V$ of size $\card{T}=k$.
As usual, a \emph{cut} is a partition of the vertices,
and its \emph{cutset} is the set of edges that connect between different parts.
Interpreting the edge weights as capacities,
the \emph{cost} of a cut $(W,V\setminus W)$ is the total weight of
the edges in the respective cutset.

We say that a cut $(W,V\setminus W)$ \emph{separates}
a terminals subset $S\subset T$ from $\bar S\eqdef T \setminus S$
(or in short that it is \emph{$S$-separating}),
if all of $S$ is on one side of the cut and $\bar S$ on the other side,
i.e., $W\cap T$ equals either $S$ or $\bar S$.
We denote by $\mincut_G(S)$ the minimum cost of an $S$-separating cut in $G$,
where by a consistent tie-breaking mechanism,
such as edge-weights perturbation,
we assume throughout that the minimum is attained by only one cut,
which we call the \emph{\MTC} (of $S$).
\begin{definition}
A network $H=(V_H,E_H,T,c_H)$ is a
\emph{cut sparsifier of $G=(V,E,T,c)$ with quality $q\ge1$ and size $s\ge k$}
(or in short, a $(q,s)$-cut-sparsifier),
if its size is $\card{V_H} \le s$ and
\begin{equation}\label{EQ:CutSparsifier}
   \forall S\subset T,\qquad
   \mincut_G(S) \leq \mincut_{H}(S) \leq q \cdot \mincut_G(S).
\end{equation}
\end{definition}
In words, \eqref{EQ:CutSparsifier} requires that every \MTC in $H$
approximates the corresponding one in $G$.
Throughout, we consider only $S\neq \emptyset,T$
although for brevity we will not write it explicitly.

Two special cases are particularly important for us.
One is quality $q=1$, or a $(1,s)$-cut-sparsifier,
which is known in the literature as a \emph{mimicking network}
and was introduced by \cite{HKNR98}.
The second case is a cut sparsifier $H$ that is furthermore a minor of $G$,
and then we call it a \emph{minor cut sparsifier},
and similarly for a \emph{minor mimicking network}.
In all our results, the sparsifier $H$ is actually a minor of $G$,
which can be important in some applications;
for instance, if $G$ is planar then $H$ admits planar-graph algorithms.

In known constructions of mimicking networks ($q=1$),
the sparsifier's size $s$ highly depends on
the number of constraints in~\eqref{EQ:CutSparsifier} that are really needed.
Naively, there are at most $2^k$ constraints, one for every \MTC
(this can be slightly optimized, e.g., by symmetry of $S$ and $\bar S$).
This naive bound was used to design, for an arbitrary network $G$,
a mimicking network whose size $s$ is exponential in the number of constraints,
namely $s\le 2^{2^k}$ \cite{HKNR98}.
A slight improvement, that is still doubly exponential in $k$,
was obtained by using the submodularity of cuts to reduce
the number of constraints \cite{KR14}.
For a planar network $G$, the mimicking network size was improved
to a polynomial in the number of constraints,
namely $s\leq k^2 2^{2k}$ \cite{KR13SODA},
and this bound is actually near-optimal, due to a very recent work
showing that some planar graphs require $s=2^{\Omega(k)}$ \cite{KPZ17}.
In this paper we explore the structure of \MTCs more deeply,
by introducing technical ideas that are new and different
from previous work like \cite{KR13SODA}.

\paragraph{Our approach.}
We take a closer look at the mimicking network size $s$ of planar graphs,
aiming at bounds that are more sensitive to the given network $G$.
For example, we would like to ``interpolate''
between the very special case of an outerplanar $G$,
which admits a mimicking network of size $s=O(k)$ \cite{CSWZ00},
and an arbitrary planar $G$ for which $s\le 2^{O(k)}$
is known and optimal \cite{KR13SODA,KPZ17}.
Our results employ a graph parameter $\gamma(G)$, defined next.

\begin{definition}[Terminal Face Cover]
The \emph{terminal face cover} $\gamma=\gamma(G)$ of a planar $k$-terminal network $G$
with a given drawing%
\footnote{We can let $\gamma$ refer to the best drawing of $G$,
and then our results might be non-algorithmic.
}
is the minimum number of faces that are incident to all the $k$ terminals,
and thus $1\leq \gamma\leq k$.
\end{definition}
This graph parameter $\gamma(G)$ is well-known to be important algorithmically.
For example, it can be used to control the runtime of
algorithms for shortest-path problems \cite{Frederickson91,CX00},
for cut problems \cite{ChenWu04,Bentz09},
and for multicommodity flow problems \cite{MNS86}.
For the complexity of computing an optimal/approximate face cover $\gamma(G)$,
see \cite{BM88,Frederickson91}.

When $\gamma=1$, all the terminals lie on the boundary of the same face,
which we may assume to be the outerface.
This special case was famously shown by Okamura and Seymour \cite{OS81}
to have a flow-cut gap of $1$ (for multicommodity flows).
Later work showed that for general $\gamma$,
the flow-cut gap is at most $3\gamma$ \cite{LS09,CSW13}.

\subsection{Main Results and Techniques}
\label{sec:results}

We provide new bounds for mimicking networks of planar graphs.
In particular, our main result refines the previous bound
so that it depends exponentially on $\gamma(G)$ rather than on $k$,
This yields much smaller mimicking networks in important cases,
for instance, when $\gamma=O(1)$ we achieve size $s=\poly(k)$.
See Table~\ref{table:mimicking} for a summary of known and new bounds.
Technically, we develop two methods to decompose the \MTCs
into ``more basic'' subsets of edges,
and then represent the constraints in~\eqref{EQ:CutSparsifier}
using these subsets.
This is equivalent to reducing the number of constraints,
and leads (as we hinted above) to a smaller sparsifier size $s$.
A key difference between the methods is that the first one in effect
restricts attention to a subset of the constraints in~\eqref{EQ:CutSparsifier},
while the second method uses alternative constraints.

\paragraph{Decomposition into elementary cutsets.}
Our first decomposition method identifies (in every graph $G$, even non-planar)
a subset of \MTCs that ``generates'' all the other ones, as follows.
First,
we call a cutset \emph{elementary}
if removing its edges disconnects the graph into exactly two
connected components (Definition~\ref{DEF: elementary cut-set}).
We then show that
every \MTC in $G$ can be decomposed into a disjoint union of elementary ones (Theorem~\ref{THM: E_S disjoint union elementary cuts}),
and use this to conclude that
if all the elementary cutsets in $G$ are well-approximated by those in $H$,
then $H$ is a cut sparsifier of $G$
(Corollary~\ref{COR: elementary cuts for general graphs}).

Combining this framework with prior work on planar sparsifier \cite{KR13SODA},
we devise the following bound that depends on $\T_e(G)$,
the set of elementary cutsets in $G$.
\begin{itemize} \compactify
\item
\textbf{Generic bound:}
Every planar graph $G$ has a mimicking network of size
$s = O(k)\cdot \card{\T_e(G)}^2$;
see Theorem~\ref{THM: planar mimick network size alpha|U|^2}.
\end{itemize}
Trivially $\card{\T_e(G)}\leq 2^k$,
and we immediately achieve $s=O(k2^{2k})$ for all planar graphs
(Corollary~\ref{COR: UB 2^2k planar graphs}).
This improves over the known bound \cite{KR13SODA} slightly (by factor $k$),
and stems directly from the restriction to elementary cutsets
(which are simple cycles in the planar-dual graph).

Using the same generic bound, we further obtain mimicking networks whose size
is \emph{polynomial} in $k$ (but inevitably exponential in $\gamma$),
starting with the base case $\gamma=1$ and then building on it, as follows.

\begin{itemize} \compactify
\item
\textbf{Base case:}
If $\gamma(G)=1$, then $\card{\T_e(G)}\leq O(k^2)$
and thus $G$ has a mimicking network of size $s=O(k^4)$;%
\footnote{The generic bound implies $s=O(k^5)$,
but we can slightly improve it in this case.}
see Theorem~\ref{THM: at most k^2 elementary min cuts}
and Corollary~\ref{THM: mimick network size k^4}.
\item
\textbf{General case, first bound:}
If $\gamma(G) \geq 1$, then $\card{\T_e(G)}\leq (2k/\gamma)^{2\gamma}$
and thus $G$ has a mimicking network of size $s=O(k(2k/\gamma)^{4\gamma})$;
see Theorem~\ref{THM: naive UB elementary cycles bounded gamma}
and Corollary~\ref{THM: naive UB mimick bounded gamma}.
\end{itemize}

The last bound on $\card{\T_e(G)}$ is clearly wasteful
(for $\gamma=k$, it is roughly quadratically worse than the trivial bound).
To avoid over-counting of edges that belong to multiple elementary cutsets,
we devise a better decomposition.

\paragraph{Further decomposition of elementary cutsets.}

Our second method decomposes each elementary cutset even further,
in a special way such that we can count the underlying fragments
(special subsets of edges) without repetitions,
and this yields our main result.

\begin{itemize} \compactify
\item
\textbf{General case, second bound.}
When $\gamma(G) \geq 1$, there are $O(2^{\gamma} k^2)$ subsets of edges,
such that every elementary cutset in $G$ can be decomposed into a disjoint union of some of these subsets.
Thus, $G$ admits a mimicking network of size $O(\gamma 2^{2\gamma}k^4)$;
see Theorem~\ref{THM: elementary mincuts bounded faces}
and Corollary~\ref{THM: UB mimicking bounded faces}.
\end{itemize}

\paragraph{Additional results.}
First, all our cut sparsifiers are also approximate flow sparsifers,
by straightforward application of the known bounds on the flow-cut gap,
see Section~\ref{sec:flow}.
Second, our decompositions easily yield
a succinct data structure that stores all the \MTCs of a planar graph $G$.
Its storage requirement depends on $\card{\T_e(G)}$,
which is bounded as above,
see Section~\ref{SEC: Terminal Cut scheme} for details.

Finally,
we show a duality between cut and distance sparsifiers (for certain graphs),
and derive new relations between their bounds, as explained next.

\begin{table}[t]
\begin{center} %
\begin{tabular}{l c l l} %
\hline
\hline        %
Graphs & Size & Minor & Reference   \\
\hline         %
General & $ 2^{2^k}\approx 2^{\binom{(k-1)}{(k-1)/2}}$ & no & \cite{HKNR98,KR14} \\
Planar & $O(k^2 2^{2k})$ & yes & \cite{KR13SODA} \\
Planar & $O(k 2^{2k})$ & yes & Corollary~\ref{COR: UB 2^2k planar graphs} \\
Planar $\gamma=\gamma(G)$ & $O(\gamma 2^{2\gamma}k^4)$ & yes & Corollary~\ref{THM: UB mimicking bounded faces} \\
Planar $\gamma(G)=1$ & $O(k^4)$ & yes & Corollary~\ref{THM: mimick network size k^4} \\
Planar $\gamma(G)=1$ & $O(k^2)$ & no & \cite{GHP17} \\
\hline  %
General & $2^{\Omega(k)}$ & no & \cite{KR13SODA,KR14} lower bound \\
Planar & $2^{\Omega(k)}$ & no & \cite{KPZ17} lower bound \\
\hline  %
\end{tabular}
\caption{Known and new bounds for mimicking networks.
\label{table:mimicking} %
}
\end{center}
\end{table}

\subsection{Cuts vs. Distances}
\label{sec:IntroDuality}

Although in several known scenarios cuts and distances are closely related,
the following notion of distance sparsification was studied separately,
with no formal connections to cut sparsifiers \cite{Gupta01,CXKR06,BG08,KNZ14,KKN15,GR16,CGH16,Cheung18,Filtser18,FKT19}.

\begin{definition}
A network $H=(V_H,E_H,T,c_H)$ is called a
\emph{$(q,s)$-distance-approximating minor (abbreviated DAM) of $G=(V,E,T,c)$},
if it is a minor of $G$, its size is $\card{V_H} \le s$ and
\begin{equation}\label{EQ: DAM}
    \forall t,t'\in T,\qquad
    d_G(t,t')\leq d_{H}(t,t') \leq q\cdot d_G(t,t'),
\end{equation}
where $d_G(\cdot, \cdot)$ is the shortest-path metric in $G$
with respect to $c(\cdot)$ as edge lengths.
\end{definition}

We emphasize that the well-known planar duality between cuts and cycles
does \emph{not} directly imply a duality between cut and distance sparsifiers.
We nevertheless do use this planar-duality approach,
but we need to break ``shortest cycles'' into ``shortest paths'',
which we achieve by adding new terminals (ideally not too many).

\begin{itemize} \compactify
\item Fix $k,q,s\ge 1$.
Then all planar $k$-terminal networks with $\gamma=1$
admit a minor $(q,s)$-cut sparsifier
if and only if all these networks admit an $(q,O(s))$-DAM;
see Theorems~\ref{THM: reduction from outer face cuts to DAM}
and \ref{THM: reduction from outer face DAM to outer face cuts}.
\end{itemize}

This result yields new cut-sparsifier bounds in the special case $\gamma=1$
(see Section~\ref{SEC: Reduction Application}).
Notice that in this case of $\gamma=1$ the flow-cut gap is $1$ \cite{OS81}, hence the three problems of minor sparsification (of distances, of cuts,
and of flows), all have the same asymptotic bounds and gaps.

This duality can be extended to general $\gamma\geq 1$ (including $\gamma=k$),
essentially at the cost of increasing the number of terminals, as follows.
If for some functions $q(\cdot)$ and $s(\cdot)$,
all planar $k$-terminal networks with given $\gamma$ admit a $(q(k),s(k))$-DAM,
then all networks in this class admits also a minor $(q(\gamma2^{\gamma}k^2),s(\gamma2^{\gamma}k^2))$-cut sparsifier.
For $\gamma=k$, we can add only $k2^k$ new terminals instead of $k^32^k$.
We omit the proof of this extension,
as applying it to the known bounds for DAM yields alternative proofs
for known/our cut-sparsifier bounds, but no new results. For example, using the reduction together with the known upper bound of $(1,k^4)$-DAM, we get that every planar $k$-terminal network with $\gamma(G)=k$ admits a minor mimicking network of size $O((k2^k)^4)$.

\paragraph{Comparison with previous techniques.}
Probably the closest notion to duality between cut sparsification and distance
sparsification is \Raecke's powerful method \cite{Raecke08},
adapted to vertex sparsification as in \cite{CLLM10,EGKRTT14,MM16}.
However, in his method the cut sparsifier $H$ is inherently randomized;
this is acceptable if $H$ contains only the terminals,
because we can take its ``expectation'' $\bar H$
(a complete graph with expected edge weights),
but it is calamitous when $H$ contains non-terminals,
and then each randomized outcome has different vertices.
Another related work, by Chen and Wu \cite{ChenWu04},
reduces multiway-cut in a planar network with $\gamma(G)=1$
to a minimum Steiner tree problem in a related graph $G'$.
Their graph transformation is similar to one of our two reductions,
although they show a reduction that goes in one direction
rather than an equivalence between two problems.

\subsection{Related Work}
\label{sec:related}

Cut and distance sparsifiers were studied extensively in recent years,
in an effort to optimize their two parameters, quality $q$ and size $s$.
The foregoing discussion is arranged by the quality parameter,
starting with $q=1$, then $q=O(1)$, %
and finally quality that grows with $k$.

\paragraph{Cut Sparsification.}
Let us start with $q=1$. Apart from the already mentioned work
on a general graph $G$ \cite{HKNR98,KR14,KR13SODA},
there are also bounds for specific graph families,
like bounded-treewidth or planar graphs \cite{CSWZ00,KR13SODA,KPZ17}.
For planar $G$ with $\gamma(G)=1$, there is
a recent tight upper bound $s=O(k^2)$ \cite{GHP17} (independent of our work),
where the sparsifier is planar but is \emph{not} a minor of the original graph.

We proceed to a constant quality $q$.
Chuzhoy~\cite{Chuzhoy12} designed an $(O(1),s)$-cut sparsifier,
where $s$ is polynomial in the total capacity incident to the terminals in
the original graph,
and certain graph families (e.g., bipartite) admit sparsifiers
with $q=1+\eps$ and $s=\poly(k/\eps)$ \cite{AGK14}.

Finally, we discuss the best quality known when $s=k$,
i.e., the sparsifier has only the terminals as vertices.
In this case, it is known that $q=O(\log k /\log \log k)$ \cite{Moitra09,LM10,CLLM10,EGKRTT14,MM16},
and there is a lower bound $q=\Omega(\sqrt{\log k})$ \cite{MM16}.
For networks that exclude a fixed minor (e.g., planar)
it is known that $q=O(1)$ \cite{EGKRTT14},
and for trees $q=2$ \cite{GR16}
(where the sparsifier is \emph{not} a minor of the original tree).

\paragraph{Distance Sparsification.}
A separate line of work studied the tradeoff between the quality $q$
and the size $s$ of a distance approximation minor (DAM).
For $q=1$, every graph admits DAM of size $s=O(k^4)$ \cite{KNZ14},
and there is a lower bound of $s=\Omega(k^2)$ even for planar graphs \cite{KNZ14}.
Independently of our work, Goranci, Henzinger and Peng \cite{GHP17}
recently constructed, for planar graphs with $\gamma(G)=1$,
a $(1,O(k^2))$-distance sparsifier
that is planar but \emph{not} a minor of the original graph.
Proceeding to quality $q=O(1)$,
planar graphs admit a DAM with $q=1+\eps$ and $s=O(k\log k/ \epsilon)^2$ \cite{CGH16},
and certain graph families, such as trees and outerplanar graphs,
admit a DAM with $q=O(1)$ and $s=O(k)$ \cite{Gupta01,BG08,CXKR06,KNZ14}.
When $s=k$ (the sparsifier has only the terminals as vertices),
then known quality is $q=O(\log k)$ for every graph \cite{KKN15,Cheung18,Filtser18}.
Additional tradeoffs and lower bounds can be found in  \cite{CXKR06,KNZ14,CGH16}.

\subsection{Preliminaries}
\label{sec:prelims}

Let $G=(V,E,T,c)$ be a $k$-terminal network,
and denote its $k$ terminals by $T=\{t_1,\ldots,t_k\}$.
We assume without loss of generality that $G$ is connected,
as otherwise we can construct a sparsifier for each connected component separately.
For every $S\subset T$, let $\mincutset_G(S)$ denote the \emph{argument}
of the minimizer in $\mincut_G(S)$, i.e.,
the minimum-cost cutset that separates $S$ from $\bar{S}=T \setminus S$ in $G$.
We assume that the minimum is unique by a perturbation of the edge weights.
Throughout, when $G$ is clear from the context, we use the shorthand
\begin{equation} \label{eq:E_S}
  E_S \eqdef \mincutset_G(S).
\end{equation}
Similarly, $CC(E_S)$ is a shorthand for the set of connected components
of the graph $G\setminus E_S$.
Define the \emph{boundary} of $W\subseteq V$, denoted $\delta(W)$,
as the set of edges with exactly one end point in $W$,
and observe that for every connected component $C\in CC(E_S)$
we have $\delta(C)\subseteq E_S$.
By symmetry, $E_S = E_{\bar{S}}$.
And since $G$ is connected and $S\neq \emptyset,T$,
we have $E_S\neq \emptyset$ and $\card{CC(E_S)}\ge 2$.
In addition, by the minimality of $E_S$,
every connected component $C\in CC(E_S)$ contains at least one terminal.

\begin{lemma}[Lemma 2.2 in \cite{KR13SODA}]
\label{LMA: every C has terminal}
For every two subsets of terminals $S,S'\subset T$ and their corresponding minimum cutsets $E_S,E_{S'}$,
every connected component $C \in CC(E_S \cup E_{S'})$ contains
at least one terminal.
\end{lemma}

\section{Elementary Cutsets in General Graphs}
\label{SEC: Elementary cuts General graphs}

In this section we define a special set of cutsets
called elementary cutsets (Definition~\ref{DEF: elementary cut-set}),
and prove that these elementary cutsets generate all other relevant cutsets,
namely, the minimum terminal cutsets in the graph
(Theorem~\ref{THM: E_S disjoint union elementary cuts}).
Therefore, to produce a cut sparsifier, it is enough to preserve only these
elementary cutsets (Corollary~\ref{COR: elementary cuts for general graphs}).
In the following discussion, we fix a network $G=(V,E,T,c)$
and employ the notations $E_S$, $CC(E_S)$ and $\delta(W)$
set up in Section~\ref{sec:prelims}.

\begin{definition}[Elementary Cutset]
\label{DEF: elementary cut-set}
Fix $S\subset T$.
Its minimum cutset $E_S$ is called an \emph{elementary cutset} if $|CC(E_S)|=2$.
\end{definition}

\begin{definition}[Elementary Component]
\label{DEF: elementary connected component}
A subset $C\subseteq V$ is called an \emph{elementary component} if $\delta(C)$ is an elementary cutset for some $S\subset T$,
i.e., $\delta(C)=E_{C\cap T}$ and $|CC(\delta(C))|=2$.
\end{definition}

Although the following two lemmas are quite straightforward,
they play a central role in the proof of Theorem~\ref{THM: E_S disjoint union elementary cuts}.

\begin{lemma}\label{LMA: boundary of C is min cutset}
Fix a subset $S\subset T$ and its minimum cutset $E_S$.
The boundary of every $C\in CC(E_S)$ is itself the minimum cutset separating the terminals $T\cap C$ from $T\setminus C$ in $G$,
i.e., $\delta(C) = E_{T\cap C}$.
\end{lemma}

\begin{proof}
    Assume toward contradiction that $\delta(C) \neq E_{T\cap C}$. Since both sets of edges separate between the terminals $T\cap C$ and $T\setminus C$, then $c(E_{T\cap C})< c(\delta(C))$. Let us replace the edges $\delta(C)$ by the edges $E_{T\cap C}$ in the cutset of $E_S$ and call this new set of edges ${E'}_S$, i.e. ${E'}_S=(E_S\setminus \delta(C))\cup E_{T\cap C}$.
    It is clear that $c({E'}_S)< c(E_S)$. We will prove that ${E'}_S$ is also a cutset that separates between $S$ and $\bar{S}$ in the graph $G$, contradicting the minimality of $E_S$.

    Assume \Wlog that $T\cap C\subseteq S$, and consider $E_S\setminus \delta(C)$. By the minimality of $E_S$ all the neighbors of $C$ contain terminals of $\bar{S}$, therefore the cutset $E_S\setminus \delta(C)$ separates the terminals $S\setminus (T\cap C)$ from $\bar{S}\cup (T\cap C)$ in $G$. Now consider $E'_S=(E_S\setminus \delta(C))\cup E_{T\cap C}$ and note that the connected component $C_{C+N_S(C)}$ contains all the terminals $T\cap C$ and some terminals of $\bar S$. This cutset $E'_S$ clearly separates $T\cap C$ from all other terminals, and also separates $S\setminus (T\cap C)$ from $\bar S \cup(T\cap C)$. Altogether this cutset separates between $S$ and $\bar{S}$ in $G$, and the lemma follows.
\end{proof}

\begin{lemma}\label{LMA: exists elementary CC}
For every $S\subset T$,
at least one component in $CC(E_S)$ is elementary.
\end{lemma}

\begin{proof}
    Fix $S\subset T$. Lemma \ref{LMA: boundary of C is min cutset} yields that $\delta(C)=E_{C\cap T}$ for every $C\in CC(E_S)$, thus it left to prove that there exists $C'\in CC(E_S)$ such that $|CC(\delta(C'))|=2$. For simplicity, we shall represent our graph $G$ as a bipartite graph $\G_S$ whose its vertices and edges are $CC(E_S)$ and $E_S$ respectively, i.e. we get $\G_S$ by contracting every $C\in CC(E_S)$ in $G$ into a vertex $v_C$. Let $V_1(\G_S)=\{v_C:\ C\cap T\subseteq S\}$ and $V_2(\G_S)=\{v_C:\ C\cap T\subseteq \bar S\}$ be the partition of $V(\G_S)$ into two sets. By the minimality of $E_S$ the graph $\G_S$ is connected, and each of $V_1(\G_S)$ and $V_2(\G_S)$ is an independent set.

    For every connected component $C\in CC(E_S)$, it is easy to see that $|CC(\delta(C))|=2$ if and only if $\G_S\setminus \{v_{C}\}$ is connected. Since $\G_S$ is connected, it has a spanning tree and thus $\G_S\setminus \{v_{C'}\}$ is connected for every leaf $v_{C'}$ of that spanning tree, and the lemma follows.
\end{proof}

\begin{theorem}[Decomposition into Elementary Cutsets]
\label{THM: E_S disjoint union elementary cuts}
For every $S\subset T$, the minimum cutset $E_S$
can be decomposed into a disjoint union of elementary cutsets.
\end{theorem}

The idea of the proof is to iteratively decrease the number of connected components in $CC(E_S)$ by uniting an elementary connected component
with all its neighbors (while recording the cutset between them),
until we are left with only one connected component --- all of $V$.

\begin{proof}[Proof of Theorem~\ref{THM: E_S disjoint union elementary cuts}]
We will need the following definition. Given $S\subset T$ and its minimum cutset $E_S$, we say that two connected components $C,C'\in CC(E_S)$ are
\emph{neighbors with respect to $E_S$}, if $E_S$ has an edge from $C$ to $C'$.
We denote by $N_S(C)\subseteq CC(E_S)$
the set of neighbors of $C$ with respect to $E_S$.
Observe that removing $\delta(C)$ from the cutset $E_S$ is equivalent to uniting the connected component $C$ with all its neighbors $N_S(C)$. Denoting this new connected component by $C_{C+N_S(C)}$ we get that $CC(E_S\setminus \delta(C))=\Big(CC(E_S)\setminus \big( \{C\}\cup N_S(C)\big) \Big)\cup \{C_{C+N_S(C)}\}$.

Let $E_S$ be a minimum cutset that separates $S$ from $\bar{S}$.
By Lemma~\ref{LMA: exists elementary CC},
there exists a component $C\in CC(E_S)$ that is elementary,
and by Lemma~\ref{LMA: boundary of C is min cutset}, $\delta(C)=E_{T\cap C}$.
Assume \Wlog that $T\cap C\subseteq S$ (rather than $\bar S$),
and unite $C$ with all its neighbors $N_S(C)$.
Now, we would like to show that this step is equivalent to ``moving''
the terminals in $C$ from $S$ to $\bar{S}$.
Clearly, the new cutset $E_S\setminus \delta(C)$
separates the terminals $S'=S\setminus (T\cap C)$
from $T\setminus {S'}=\bar{S}\cup (T\cap C)$,
but to prove that
\begin{equation} \label{eq:E_S'}
  E_{S'}=E_S\setminus \delta(C),
\end{equation}
we need to argue that this new cutset %
has \emph{minimum cost} among those separating $S'$ from $\bar{S}'$.
To this end, assume to the contrary;
then $E_{S'}$ must have a strictly smaller cost than $E_S\setminus \delta(C)$, because both cutsets separate $S'$ from $T\setminus {S'}$.
Now similarly to the proof of Lemma~\ref{LMA: boundary of C is min cutset},
it follows that $E_{S'}\cup \delta(C)$ separates $S$ from $\bar{S}$,
and has a strictly smaller cost than $E_S$,
which contradicts the minimality of $E_S$.

Using \eqref{eq:E_S'}, we can write $E_{S}=\delta(C)\cupdot E_{S'}$
and continue iteratively with $E_{S'}$ while it is non-empty
(i.e., $|CC(E_{S'})|>1$).
Formally, the theorem follows by induction on $|CC(E_S)|$.
\end{proof}

To easily examine all the elementary cutsets in a graph $G$,
we define
\begin{linenomath}
$$ \T_e(G) \eqdef \{S\subset T \ :\ |CC(E_S)|=2\}. $$
\end{linenomath}
Using Theorem~\ref{THM: E_S disjoint union elementary cuts},
the cost of every minimum terminal cut can be recovered, in a certain manner,
from the costs of the elementary cutsets of $G$,
and this yields the following corollary.

\begin{corollary}\label{COR: elementary cuts for general graphs}
Let $H$ be a $k$-terminal network with same terminals as $G$.
If %
$\T_e(G)=\T_e(H)$ and
\begin{equation} \label{EQ: elementary cut sparsifier}
  \forall S\in \T_e(G),\qquad
  \mincut_G(S) \leq \mincut_H(S)\leq q\cdot \mincut_G(S) ,
\end{equation}
then $H$ is a cut-sparsifier of $G$ of quality $q$.
\end{corollary}

\begin{proof}
Given $G$ and $H$ as above, we only need to prove~\eqref{EQ:CutSparsifier}.
To this end, fix $S\subset T$.
Observe that for every $\varphi\subseteq \T_e(G)=\T_e(H)$, the set $\mycup_{S'\in \varphi}\mincutset_G(S')$ is $S$-separating in $G$ if and only if for every $t\in S$ and $t'\in \bar S$ there exists $S'\in \varphi$ such that \Wlog $t\in S'$ and $t'\notin S'$. Thus, if $\varphi$ is a partition of $S$, i.e. $S=\cupdot\limits_{S'\in \varphi}S'$, then $\mycup_{S'\in \varphi}\mincutset_G(S')$ is $S$-separating in $G$. Since the same arguments hold also for $H$, we get the following:
\begin{linenomath}
\begin{equation*} %
  \mycup_{S'\in \varphi}\mincutset_G(S')\text{ is } S\text{-separating in } G
  \Leftrightarrow
  \mycup_{S'\in \varphi}\mincutset_H(S')\text{ is } S\text{-separating in } H.
\end{equation*}
\end{linenomath}

By Theorem~\ref{THM: E_S disjoint union elementary cuts},
there exists $\varphi_G\subseteq \T_e(G)$ such that $S=\cupdot\limits_{S'\in \varphi_G}S'$ and\\
$\mincutset_G(S)=\cupdot\limits_{S'\in \varphi_G}\mincutset_G(S')$, thus
\[
  \mincut_H(S)
  \leq\sum_{S'\in \varphi_G}\mincut_H(S')
  \leq q\cdot \sum_{S'\in \varphi_G} \mincut_G(S')
  = q\cdot \mincut_G(S).
\]
Applying Theorem~\ref{THM: E_S disjoint union elementary cuts} to $H$ together with an analogous argument yields that $\mincut_G(S) \leq \mincut_H(S)$, which proves \eqref{EQ:CutSparsifier} and the corollary follows.
\end{proof}

\section{Mimicking Networks for Planar Graphs}
\label{SEC: UB mimicking planar graphs}
We now present an application of our results in Section~\ref{SEC: Elementary cuts General graphs}.
We begin with a bound on the mimicking network size for a planar graph $G$
as a function of the number of elementary cutsets
(Theorem~\ref{THM: planar mimick network size alpha|U|^2}).
We then obtain an upper bound of $O(k2^{2k})$ for every planar network
(Corollary~\ref{COR: UB 2^2k planar graphs}),
which improves the previous work \cite[Theorem 1.1]{KR13SODA}
by a factor of $k$, thanks to the use of elementary cuts.
The underlying reason is that the previous analysis in~\cite{KR13SODA}
considers all the $2^k$ possible terminal cutsets, and each of them
is a collection of at most $k$ simple cycles in the dual graph $G^*$.
We can consider only the elementary cutsets
by Corollary~\ref{COR: elementary cuts for general graphs},
and each of them is a simple cycle in $G^*$
by Definition~\ref{DEF: elementary cut-set}.
Thus, we consider a total of $2^k$ simple cycles,
saving a factor of $k$ over the earlier naive bound of $k2^k$ simple cycles.

\begin{theorem}\label{THM: planar mimick network size alpha|U|^2}
Every planar network $G$,
in which $|CC(E_S\cup E_{S'})|\leq \alpha$ for all $S,S'\in \T_e(G)$,
admits a minor mimicking network $H$ of size $O(\alpha\cdot |\T_e(G)|^2)$.
\end{theorem}
The proof of this theorem appears in Section~\ref{Sec: UB mimicking planar}.
It is based on applying the machinery of \cite{KR13SODA},
but restricting the analysis to elementary cutsets.
\begin{corollary}\label{COR: UB 2^2k planar graphs}
Every planar network $G$ admits a minor mimicking network of size $O(k2^{2k})$.
\end{corollary}
\begin{proof}
Apply Theorem~\ref{THM: planar mimick network size alpha|U|^2},
using an easy bound $\alpha = O(k)$ from Lemma~\ref{LMA: every C has terminal},
and a trivial bound on the number of elementary cutsets
$|\T_e(G)| \leq 2^k$.
\end{proof}

\subsection{Proof of Theorem~\ref{THM: planar mimick network size alpha|U|^2}}
\label{Sec: UB mimicking planar}

Given a $k$-terminal network $G$ and $\alpha>0$ such that $|CC(E_S\cup E_{S'})|\leq \alpha$ for every $S,S'\in \T_e(G)$, we prove that it admits a minor mimicking network $H$ of size $O(\alpha|\T_e(G)|^2)$. Let $\hat E=\cup_{S\in \T_e(G)}E_S$, and construct $H$ by contracting every connected component of $G\setminus \hat E$ into a single vertex. Notice that edge contractions can only increase the cost of any minimum terminal cut, and that in our construction edges of an elementary cutset of $G$ are never contracted.
Thus, the resulting $H$ is a minor of $G$, that maintains all the elementary cutsets of $G$, and by Corollary~\ref{COR: elementary cuts for general graphs} $H$ maintains all the terminal mincuts of $G$. We proceed to bound the number of connected components in $G\setminus \hat E$, as this will clearly be the size of our mimicking network $H$. The crucial step here is to use the planarity of $G$ by employing the dual graph of $G$ denoted by $G^*$ (for basic notions of planar duality see Appendix~\ref{app: Planar Duality}).

    Loosely speaking, the elementary cutsets in $G$ correspond to cycles in the dual graph $G^*$,
    and thus we consider the dual edges of $\hat E$,
    which may be viewed as a subgraph of $G^*$ comprising of (many) cycles.
    We then use Euler's formula and the special structure of this subgraph of cycles;
    more specifically, we count its meeting vertices, which turns out to
    require the aforementioned bound of $\alpha$ for two sets of terminals $S,S'$.
    This gives us a bound on the number of faces in this subgraph,
    which in turn is exactly the number of connected components
    in the primal graph (Lemma \ref{LMA: bound number of faces in dual graph using meeting points}). Observe that removing edges from a graph $G$ can disconnect it into
    (one or more) connected components.
    The next lemma characterizes this behavior in terms
    of the dual graph $G^*$. Let $V_m(G)$ be all the vertices in the graph $G$ with degree $\geq 3$, and call them \emph{meeting vertices} of $G$. The following lemma bounds the number of meeting vertices in two elementary cuts by $O(\alpha)$.

    \begin{lemma}\label{LMA: meeting vertices in two elementary cycles}
    For every two subsets of terminals $S,S'\in \T_e(G)$, the dual graph $G^*[E_S^*\cup E_{S'}^*]$ has at most $2\alpha$ meeting vertices.
    \end{lemma}

    \begin{proof}[Proof Sketch]
        For simplicity denote by $G^*_{SS'}$ the graph $G^*[E_S^*\cup E_{S'}^*]$.
        By our assumption, the graph $G\setminus (E_S\cup E_{S'})$ has at most $\alpha$ connected components. By Lemma \ref{LMA: connected components vs regions} every connected component in $G\setminus(E_S \cup E_{S'})$ corresponds to a face in $G^*_{SS'}$. Therefore, $G^*_{SS'}$ has at most $\alpha$ faces. Let $V_{SS'},E_{SS'}$ and $F_{SS'}$ be the vertices, edges and faces of the graph $G^*_{SS'}$. Note that the degree of every vertex in that graph is at least 2. Thus, by the degree-sum formula (the total degree of all vertices equals to twice the number of edges), $2|E_{SS'}|\geq 2|V_{SS'}\setminus V_m( G^*_{SS'})|+3|V_m( G^*_{SS'})|$ and so $|E_{SS'}|\geq |V_{SS'}|+\frac{1}{2}V_m( G^*_{SS'})$. Together with Euler formula we get that $\alpha\geq |F_{SS'}|\geq |E_{SS'}|-|V_{SS'}|\geq \frac{1}{2}|V_m(G^*_{SS'})|$, and the lemma follows.
    \end{proof}

\begin{lemma}\label{LMA: bound number of faces in dual graph using meeting points}
The dual graph $G^*[\hat{E}^*]$ has at most $O(\alpha|\T_e(G)|^2)$ faces. Thus, $G\setminus \hat{E}$ has at most $O(\alpha|\T_e(G)|^2)$ connected components.
\end{lemma}

\begin{proof}[Proof Sketch]
For simplicity denote by $\hat G^*$ the graph $G^*[\hat E^*]$, and let $E_m(\hat G^*)$ be all the edges in $\hat G^*$ that are incident to meeting vertices. Fix an elementary subset of terminals $S\in \T_e(G)$. By Lemma \ref{LMA: meeting vertices in two elementary cycles} there are at most $2\alpha$ meeting vertices in $G^*[E_S^*\cup E_{S'}^*]$, for every $S'\in \T_e(G)$. Summing over all the different $S'$ in $\T_e(G)$ we get that there are at most $2\alpha|\T_e(G)|$ meeting vertices on the cycle $E^*_S$ in the graph $\hat G^*$. Since the degree of every vertex in $G^*(E_S^*)$ is 2, we get that
        \begin{linenomath}
        $$|E_S^*\cap E_m(\hat G^*)|\leq 2|V(G^*(E_S^*))\cap V_m(\hat G^*)|\leq 4\alpha |\T_e(G)|.$$
        \end{linenomath}
        Again summing over at most $|\T_e(G)|$ different elementary subsets $S$ we get that $|E_m(\hat G^*)|\leq 4\alpha|\T_e(G)|^2$. Plugging it into Euler formula for the graph $G^*[\hat E^*]$, together with the inequality $|E(\hat G^*)\setminus E_m(\hat G^*)|\leq |V(\hat G^*)\setminus V_m( \hat G^*)|$ by the fact that the two sides represent the edges and vertices of a graph consisting of vertex-disjoint paths (because its maximum degree is at most $2$), we get the following

\begin{linenomath}
        \begin{align*}
            |F(\hat G^*)|&=|E(\hat G^*)|-|V(\hat G^*)|+1+|CC(\hat G^*)|\\
            &\leq |E_m(\hat G^*)|-|V_m(\hat G^*)|+1+|CC(\hat G^*)|\\
            &\leq 4\alpha |\T_e(G)|^2+1+|CC(\hat G^*)|.
        \end{align*}
\end{linenomath}

        Since $|\T_e(G)|\geq k$ it left to bound $|CC(\hat G^*)|$ by $k$. Assume towards contradiction that $|CC(\hat G^*)|\geq k+1$, thus there exists a connected component $W$ in $\hat G^*$ that does not contains a terminal face of $G^*$. By the construction of $\hat E^*$, $W$ contains at least one elementary shortest cycle that separates between terminal faces of $G^*$ in contradiction. Finally, Lemma \ref{LMA: connected components vs regions} with $M=\hat{E}$ yields that $|CC(G\setminus \hat{E})|=|F(G^*[\hat E^*])|=O(\alpha|\T_e(G)|^2)$ and the lemma follows.
    \end{proof}

    Recall that we construct our mimicking network $H$ by contracting every connected component of $G\setminus \hat{E}$ into a single vertex. By Lemma \ref{LMA: bound number of faces in dual graph using meeting points} we get that $H$ is a minor of $G$ of size $O(\alpha|\T_e(G)|^2)$ and Theorem~\ref{THM: planar mimick network size alpha|U|^2} follows.

\section{Mimicking Networks for Planar Graphs with Bounded $\gamma(G)$}
\label{SEC: UB mimicking planar graphs bounded gamma}

In this section,
the setup is that $G=(V,E,T,c)$ is a planar $k$-terminal network with terminal face cover $\gamma=\gamma(G)$.
Let $f_1,\ldots, f_{\gamma}$ be faces
that are incident to all the terminals,
and let $k_i$ denote the number of terminals incident to face $f_i$.
We can in effect assume that $\sum_{i=1}^{\gamma} k_i = k$,
because we can count each terminal as incident to only one face, and ``ignore'' its incidence to the other $\gamma-1$ faces (if any). %

Our goal is to construct for $G$ a mimicking network $H$,
and bound its size as a function of $k$ and $\gamma(G)$.
Our construction of $H$ is the same as
in Theorem~\ref{THM: planar mimick network size alpha|U|^2},
and the challenge is to bound its size. The implications to flow sparsifiers are discussed in Section~\ref{sec:flow}.

\paragraph{All terminals are on one face.}
We start with the basic case $\gamma=1$,
i.e., all the terminals are on the same face,
which we can assume to be the outerface.
The idea is to apply Theorem~\ref{THM: planar mimick network size alpha|U|^2}.
The first step is to characterize all the elementary cutsets,
which yields immediately an upper bound on their number.
The second step is to analyze the interaction between any two elementary cutsets.

\begin{theorem}\label{THM: at most k^2 elementary min cuts}
In every planar $k$-terminal network $G$ with $\gamma(G)=1$,
the number of elementary cutsets is
$\card{\T_e(G)} \leq \binom{k}{2}$.
\end{theorem}

The proof,
appearing in Section~\ref{SEC: Elementary cuts planar graphs 1 face},
is based on two observations that view the outerface as a cycle of vertices:
(1) every elementary cutset disconnects the outerface's cycle into two paths,
which we call intervals (see Definition~\ref{Def: Interval}); and
(2) every such interval can be identified by the terminals it contains.
It then follows that every elementary cutset $E_S$ is uniquely determined
by two terminals, leading to the required bound.

The next lemma bounds the interaction between any two elementary cutsets. Its proof appears at the end of Section~\ref{SEC: Elementary cuts planar graphs 1 face}.

\begin{lemma}\label{LMA: Connected components two min-cuts}
For every planar $k$-terminal network $G$ with $\gamma(G)=1$, and for every $S,S'\in \T_e(G)$, there are at most $4$ connected components in $G \setminus (E_S \cup E_{S'})$.
\end{lemma}

\begin{corollary}\label{THM: mimick network size k^4}
Every planar $k$-terminal network $G$ with $\gamma(G)=1$
admits a minor mimicking network of size $s=O(k^4)$.
\end{corollary}

\begin{proof}
Apply Theorem~\ref{THM: planar mimick network size alpha|U|^2},
with $\alpha = 4$ from Lemma~\ref{LMA: Connected components two min-cuts}
and $|\T_e(G)|\leq k^2$
from Theorem~\ref{THM: at most k^2 elementary min cuts}.
\end{proof}

\paragraph{All terminals are on $\gamma$ faces, first bound.}

Our first (and weaker) bound for the general case $\gamma \geq 1$
follows by applying Theorem~\ref{THM: planar mimick network size alpha|U|^2}.
To this end, we bound the number of elementary cutsets by $(2k/\gamma)^{2\gamma}$
in Theorem~\ref{THM: naive UB elementary cycles bounded gamma},
whose proof is in Section~\ref{SEC: naiv UB elementary cuts bounded gamma},
and then conclude a mimicking network size of $O(k(2k/\gamma)^{4\gamma})$
in Corollary~\ref{THM: naive UB mimick bounded gamma}.

\begin{theorem}\label{THM: naive UB elementary cycles bounded gamma}
In every planar $k$-terminal network $G$ with $\gamma=\gamma(G)$,
the number of different elementary cutsets is
$\card{\T_e(G)}
  \leq 2^{2\gamma} \left(\Pi_{i=1}^{\gamma}k_i^2 \right)
  \leq (2k/\gamma)^{2\gamma}$.
\end{theorem}

\begin{corollary}\label{THM: naive UB mimick bounded gamma}
Every planar $k$-terminal network $G$ with $\gamma=\gamma(G)$ admits a mimicking network of size $s= O\left(k(2k/\gamma)^{4\gamma}\right)$.
\end{corollary}

\begin{proof}
Apply Theorem~\ref{THM: planar mimick network size alpha|U|^2},
using $\alpha = O(k)$ from Lemma~\ref{LMA: every C has terminal}
and $|\T_e(G)| \leq (2k/\gamma)^{2\gamma}$
from Theorem~\ref{THM: naive UB elementary cycles bounded gamma}.
\end{proof}

\paragraph{All terminals are on $\gamma$ faces, second bound.}

Our second (and improved) result for the general case $\gamma \geq 1$
follows by a refined analysis of the elementary cutsets.
While our bound of $(2k)^{2\gamma}$ on the number of elementary cutsets is tight,
it leads to a wasteful mimicking network size
(for example, plugging the worst-case $\gamma=k$
into Corollary~\ref{THM: naive UB mimick bounded gamma}
is inferior to the bound in Corollary~\ref{COR: UB 2^2k planar graphs}).
The reason is that this approach over-counts edges of the mimicking network,
and we therefore devise a new proof strategy that decomposes each elementary cutset even further, in a special way that lets us to count the underlying fragments (special subsets of edges) without repetitions.
We remark that the actual proof works in the dual graph $G^*$,
and decomposes a simple cycle into (special) paths.

\begin{theorem}[Further Decomposition of Elementary Cutsets]
\label{THM: elementary mincuts bounded faces}
Every planar $k$-terminal network $G$ with ${k_1,\ldots,k_\gamma}$ as above,
has
$p = 2^{\gamma}\big( 1+\sum_{i,j=1}^{\gamma} k_i k_j \big)
   \leq O(2^{\gamma}k^2)$
subsets of edges $E_1,\ldots,E_p\subset E$,
such that every elementary cutset in $G$
can be decomposed into a disjoint union of some of these $E_i$'s, and each of $E_i$ contains exactly 2 edges from the boundaries of the faces $f_1,\ldots,f_{\gamma}$.
\end{theorem}

We prove this theorem in Section~\ref{SEC: proof elementary mincuts bounded gamma}.
The main difficulty is to define subsets of edges
that are contained in elementary cutsets and are also easy to identify.
We implement this identification by attaching to every such subset
a three-part label.
We prove that each label is unique,
and count the number of different possible labels,
which obviously bounds the number of such ``special'' subsets of edges.

\begin{corollary} \label{THM: UB mimicking bounded faces}
Every planar $k$-terminal network $G$ with $\gamma=\gamma(G)$
admits a minor mimicking network of size $s=O(\gamma 2^{2\gamma}k^4)$.
\end{corollary}

A slightly weaker bound of $O(2^{2\gamma}k^5)$ on the mimicking network size
follows easily from Theorem~\ref{THM: planar mimick network size alpha|U|^2}
by replacing elementary cutsets with our ``special'' subsets of edges.
To this end, it is easy to verify that all arguments about elementary cutsets
hold also for the ``special'' subsets of edges.
This includes the bound $\alpha=O(k)$,
because if every two elementary cutsets intersect at most $O(k)$ times,
then certainly every two ``special'' subsets (which are subsets of elementary cutsets) intersect at most $O(k)$ times.
We can thus apply Theorem~\ref{THM: planar mimick network size alpha|U|^2}
with $\alpha=O(k)$ and ``replacing'' $|\T_e(G)|$ with $O(2^{\gamma}k^2)$
that we have by Theorem~\ref{THM: elementary mincuts bounded faces}.
The stronger bound in Corollary~\ref{THM: UB mimicking bounded faces}
follows by showing that $\alpha=O(\gamma)$ for ``special'' subsets of edges.

\begin{proof}[Proof of Corollary~\ref{THM: UB mimicking bounded faces}]
Given a planar $k$-terminal network $G$ with $\gamma(G)=\gamma$,
use Theorem~\ref{THM: elementary mincuts bounded faces} to decompose
the elementary terminal cutsets of $G$ into $p=O(2^{\gamma}k^2)$ subsets of edges $E_1,\ldots, E_p$ as stated above. Since each of $E_i$ has exactly two edges from the boundaries of the faces $f_1,\ldots, f_{\gamma}$, then for every $E_i$ and $E_j$ there are at most $O(\gamma)$ connected components in $G\setminus(E_i\cup E_j)$ that contain terminals. Let $E_S$ and $E_{S'}$ be elementary cutsets of $G$ such that $E_i\subseteq E_S$ and $E_j\subseteq E_{S'}$. By Lemma~\ref{LMA: every C has terminal},
each connected component in $G\setminus(E_S\cup E_{S'})$ must contain
at least one terminal. Thus, each connected component in $G\setminus(E_i\cup E_{j})$ must contain
at least one terminal, which bound the number of its connected components by $O(\gamma)$. Apply Theorem~\ref{THM: planar mimick network size alpha|U|^2} with $|\T_e(G)|=O(2^{\gamma}k^2)$ and $\alpha=O(\gamma)$ and the corollary follows.
\end{proof}

\subsection{Proof of Theorem \ref{THM: at most k^2 elementary min cuts} and Lemma~\ref{LMA: Connected components two min-cuts}}
\label{SEC: Elementary cuts planar graphs 1 face}

In this section we prove Theorem~\ref{THM: at most k^2 elementary min cuts},
which bounds the number of elementary cutsets when $\gamma=1$. We start with a few definitions and lemmas. Let $G=(V,E,T,c)$ be a connected planar $k$-terminal network, such that the terminals $t_1,\ldots,t_k$ are all on the same face in that order. Assume \Wlog that this special face is the outerface $f_\infty$. We refer to this outerface as a clockwise-ordered cycle $\langle v^\infty_1,v^\infty_2, \ldots, v^\infty_{l} \rangle$, such that for every two terminals $t_i,t_j$ if $v^\infty_x=t_i$ and $v^\infty_y=t_j$ then $i<j$ if and only if $x<y$.

\begin{definition}\label{Def: Interval}
An \emph{interval} of $f_\infty$ is a subpath $\hat I = \langle v^\infty_{i}, v^\infty_{i+1}, \ldots, v^\infty_{j} \rangle$ if $i\leq j$ and in the case where $i>j$ $\hat I = \langle v^\infty_{i}, \ldots,v^\infty_{l},v^\infty_{1}, \ldots, v^\infty_{j}\rangle$. Denote its vertices by $V(\hat I)$ or, slightly abusing notation, simply by $\hat I$.
\end{definition}

Two trivial cases are a single vertex $\langle v^\infty_i \rangle$ if $i=j$, and the entire outerface cycle $f_\infty$ if $i=j-1$.

\begin{definition}\label{DEF: Maximal Order Partition Interval wrt W}
Given $W\subseteq V$, an interval $\hat I=\langle v^\infty_{i}, v^\infty_{i+1}, \ldots, v^\infty_{j}\rangle$ is called \emph{maximal with respect to $W$}, if $V(\hat I)\subseteq W$ and no interval in $W$ strictly contains $\hat I$, i.e. $v^\infty_{i-1},v^\infty_{j+1}\notin W$. Let $\I(W)$ be the set of all maximal intervals with respect to $W$, and let the \emph{order} of $W\subseteq V$ be $|\I(W)|$.
\end{definition}

Observe that $\I(W)$ is a unique partition of $W\cap V(f_\infty)$, hence the order of $W$ is well defined. Later on, we apply Definition \ref{DEF: Maximal Order Partition Interval wrt W} to connected components $C\in CC(E_S)$, instead of arbitrary subsets $W\subseteq V$.
For example, in Figure \ref{FGR: General E_S and CC(E_S) in G}, $\I(C_3)=\{I_3,I_5,I_9\}$, and the order of $C_3$ is $|\I(C_3)|=3$.

\begin{lemma}\label{LMA: E_S partition outerface to intervals}
For every subset $S\subset T$, $\cupdot_{C\in CC(E_S)}\I(C)$ is a partition of $V(f_\infty)$.
\end{lemma}

\begin{proof}
Fix $S\subset T$ and its minimum cutset $E_S$. Then $CC(E_S)$ is a partition of the vertices $V$ into connected components. It induces a partition also of $V(f_\infty)$, i.e $V(f_\infty)=\cupdot_{C\in CC(E_S)} \big( C\cap V(f_\infty) \big)$. By Definition \ref{DEF: Maximal Order Partition Interval wrt W}, each $C\cap V(f_\infty)$ can be further partitioned into maximal intervals, given by $\I(C)$. Combining all these partitions, and the lemma follows.
See Figure \ref{FGR: General E_S and CC(E_S) in G} for illustration.
\end{proof}

\begin{figure}
  \centering
  \includegraphics[angle=0,width=0.6\textwidth]{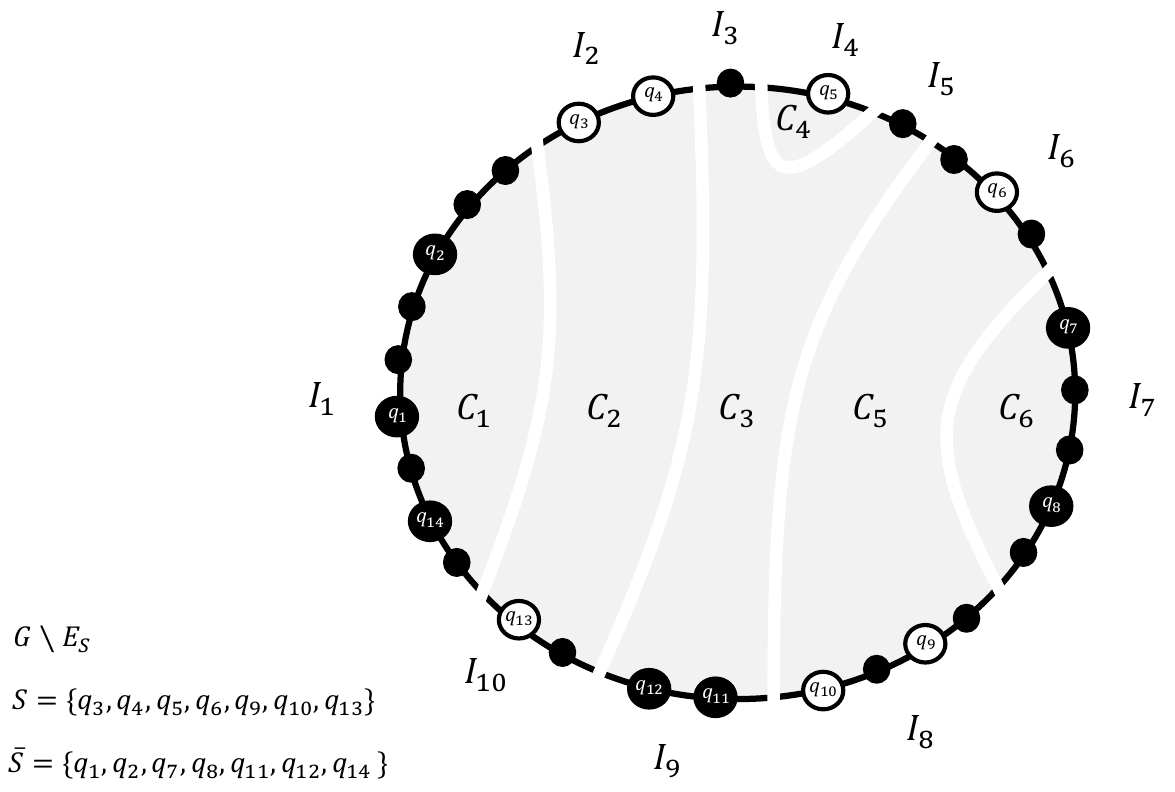}
  \caption{A minimum cutset $E_S$ that partitions the graph $G$ into $6$ connected components $CC(E_S)=\{C_1,\ldots, C_6\}$, and partitions the outerface vertices into $10$ intervals $\I(CC(E_S))=\{I_1, \ldots, I_{10}\}$.
Large nodes represent terminals,
colored according to whether they lie in $S$ or $\bar{S}$.}
  \label{FGR: General E_S and CC(E_S) in G}
  \mbox{}\hrule
\end{figure}

\begin{lemma}\label{LMA: if C elementary then I(C)=1}
For every $S\subset T$ and its minimum cutset $E_S$,
if $C\in CC(E_S)$ is an elementary connected component in $G$, then $|I(C)|=1$.
\end{lemma}

\begin{proof}
Since $C$ is elementary, there are exactly two connected components $C$ and $C'$ in $CC(\delta(C))$.
By Lemma \ref{LMA: every C has terminal} each of $C$ and $C'$ contains
at least one terminal. Since all the terminals are on the outerface,
each of $C$ and $C'$ contains at least one interval. Assume toward
contradiction that $C$ contains at least two maximal intervals $I_1$
and $I_3$, then there must be at least two intervals $I_2$ and $I_4$
in $C'$ that appear on the outerface in an alternating order,
i.e. $I_1,I_2,I_3,I_4$. Let $v_i$ be a vertex in the interval $I_i$,
and denote by $P_{13}$ and $P_{24}$ a path that connects between
$v_1,v_3$ and between $v_2,v_4$ correspondingly. Note that $P_{13}$ is
contained in $C$ and $P_{24}$ is contained in $C'$. Moreover note that
these two paths must intersect each other, giving a contradiction, and the
lemma follows.
\end{proof}

We are ready to prove Theorem~\ref{THM: at most k^2 elementary min cuts}. Recall that for every $S\subset T$, if $E_S$ is an elementary cutset then by Lemma \ref{LMA: if C elementary then I(C)=1} each of $S$ and $\bar S$ must be a single interval.
Hence they must be of the form $\{{t_i}, {t_{i+1}},\ldots, {t_j}\}$ and $\{{t_{j+1}},\ldots, {t_{i-1}}\}$ allowing wrap-around.
Thus, we can characterize $S$ and $\bar S$ by the pairs $(i,j)$ and $(j+1,i-1)$ respectively. There are at most $k(k-1)$ such different pairs, since $S\neq T$ and thus $j\neq i-1$. By the symmetry between $S$ and $\bar S$,
we should divide that number by $2$ and Theorem~\ref{THM: at most k^2 elementary min cuts} follows.

\begin{proof}[Proof of Lemma~\ref{LMA: Connected components two min-cuts}]
    Let $E_S$ and $E_{S'}$ be two elementary minimum cutsets, and let $C_S$ and $C_{\bar{S}}$ be the two elementary connected components in $CC(E_S)$. By Lemma \ref{LMA: if C elementary then I(C)=1}, each of $C_S$ and $C_{\bar S}$ contains exactly one maximal interval denoted by $I_S$ and $I_{\bar{S}}$ respectively, and similarly denote $C_{S'}, C_{\bar S'}, I_{S'}$ and $I_{\bar S'}$ for $E_{S'}$. Since each of the cutsets $E_S$ and $E_{S'}$ intersect the cycle of the outerface in exactly two edges, the cutset $E_S \cup E_{S'}$ intersects the cycle of the outerface in at most $4$ edges. Therefore the graph $G\setminus (E_S \cup E_{S'})$ has at most $4$ maximal intervals. By Lemma \ref{LMA: every C has terminal}, every connected component in $CC(E_S \cup E_{S'})$ must contains at least one terminal. Since all the terminals lie on the outerface, any connected component that contains terminal must contains also an interval. Every interval is contained in exactly one connected component.
    Thus, there are at most $4$ connected components in $CC(E_S \cup E_{S'})$, and the lemma follows.
    See Figure~\ref{FGR: PartitionTwoIntervals} for illustration.
\end{proof}

\begin{figure}
        \centering
        \includegraphics[angle=0,width=0.9\textwidth]{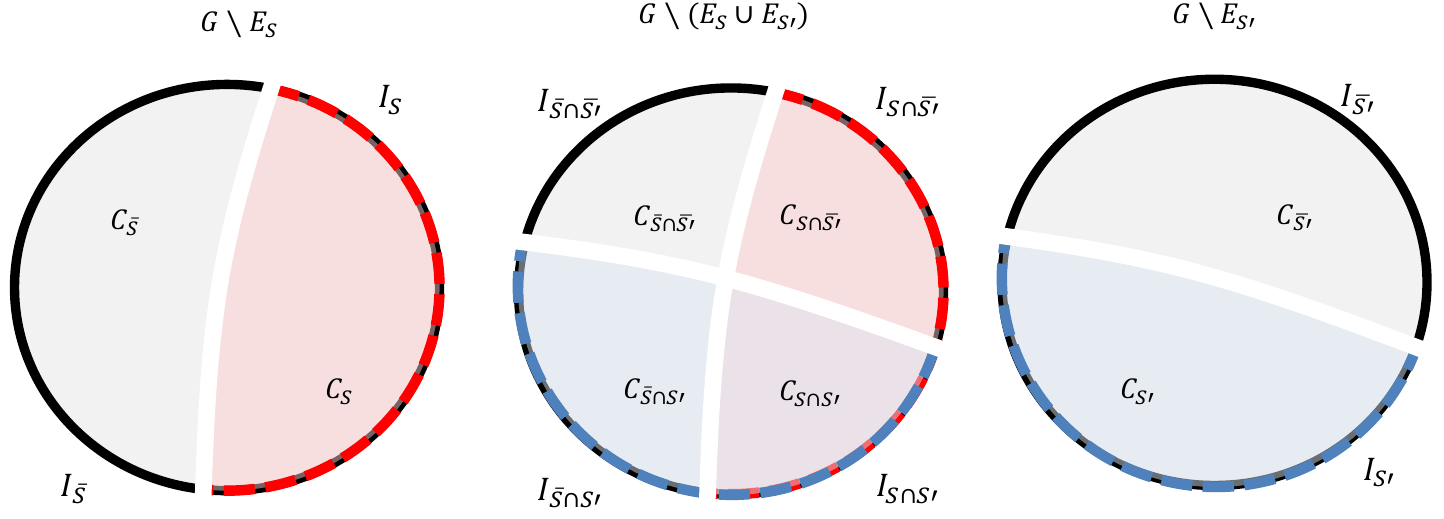}
        \caption{The union of two elementary cutsets, $E_S\cup E_{S'}$, disconnects $G$ into (at most) $4$ connected components,
and the outerface into 4 intervals.}
        \label{FGR: PartitionTwoIntervals}
        \mbox{}\hrule
\end{figure}

\subsection{Proof of Theorem~\ref{THM: naive UB elementary cycles bounded gamma}}
\label{SEC: naiv UB elementary cuts bounded gamma}

In this section we prove Theorem~\ref{THM: naive UB elementary cycles bounded gamma},
which bounds the number of elementary cutsets when $\gamma>1$. Since we assume (by perturbation) that there is a one-to-one correspondence between $S\subseteq T$ and $E_S$, it suffices to bound the number of different ways that an elementary cutset can partition the terminals into $S$ and $\bar S$. We achieve the latter by two observations, which are extensions of the ideas in Theorem~\ref{THM: at most k^2 elementary min cuts}. First, an elementary cutset can break each of the $\gamma$ faces into at most two paths, which overall splits the terminals into at most $2\gamma$ subsets. As each subset (path) can lie either in $S$ or in $\bar S$, there are at most $2^{2\gamma}$ different ways to partition $T$ into $S$ and $\bar S$ (this bound includes cases where two paths from the same face lie both in $S$ or both in $\bar S$, which is equivalent to not breaking the face into two paths). Second, there are $\binom{k_i}{2}$ ways that the face $f_i$ can be broken into 2 paths by elementary cutsets, which gives overall $\Pi_{i=1}^{\gamma}\binom{k_i}{2}$ ways to break all the $\gamma$ faces simultaneously. Combining these two observations leads to the required bound.

We start with a few definitions. Let $G=(V,E,F,T,c)$ be a $k$-terminal network with $\gamma$ faces $f_1,\ldots, f_{\gamma}$, where each $f_i$ contains the $k_i$ terminals $T_i$ (breaking ties arbitrarily),
where $T=\cupdot_{i=1}^{\gamma} T_i$ and thus $k=\sum_{i=1}^{\gamma}k_i$.
Denote the terminals in $T_i$ by $t^i_{1},\ldots,t^i_{k_i}$,
where the order is by a clockwise order around the boundary of $f_i$,
starting with an arbitrary terminal; for simplicity, we shall write $t_j$ instead of $t^i_j$ when the face $f_i$ is clear from the context.
Let $G^*$ be the dual graph of $G$.
The graph $G^*$ has $k$ terminal faces $\{f_{t^i_j}\}$
that are dual to the terminals $\{t^i_j\}$ of $G$,
and has $\gamma$ special vertices $W=\{w_1, \ldots, w_{\gamma}\}$
that are dual to the faces $f_1,\ldots,f_{\gamma}$ of $G$ (see Appendix~\ref{app: Planar Duality} for basic notions of planar duality).

We label each $S\in \T_e(G)$ (and its elementary cycle $E^*_S$)
by two vectors $\bar x,\bar y$, as follows.
Since $E^*_S$ is a simple cycle, it visits every vertex $w_i\in W$ at most once.
If it does visit $w_i$, then exactly two cycle edges are incident to $w_i$.
and these two edges naturally partition the faces around $w_i$ into two subsets.
Moreover, each subset appears as a contiguous subsequence
if the faces around $w_i$ are scanned in a clockwise order.
In particular, the terminal faces $f_{t_1},\ldots,f_{t_{k_i}}$
are partitioned into two subsets, whose indices can be written as
$\{t_{x_i},\ldots,t_{y_i-1}\}$ and $\{t_{y_i},\ldots,t_{x_i-1}\}$,
for some $x_i,y_i\in[k_i]$, under the two conventions:
(i) we allow wraparound, i.e., $t_{k_i+1}=t_1$ and so forth;
(ii) if $x_i=y_i$, then we have a trivial partition of $T_i$,
where one subset is $T_i$ and the other is $\emptyset$.
Observe that one of these subsets is contained in $S$ and the other in $\bar S$,
thus we can assume that $\{t_{x_i},\ldots,t_{y_i-1}\}\subseteq S$
and $\{t_{y_i},\ldots,t_{x_i-1}\}\subset \bar S$.
If the cycle $E^*_S$ does not visit $w_i$, then we simply define $x_i=y_i=1$,
which represents a trivial partitioning of $T_i$.
The labels are now defined as
$\bar x=(x_1,\ldots,x_{\gamma})$ and $\bar y=(y_1,\ldots,y_{\gamma})$.

We now claim that $G^*$ has at most $2^{2\gamma}$ elementary cycles
with the same label $(\bar x, \bar y)$.
To see this, fix $\bar x,\bar y \in [k_1]\times\cdots\times[k_{\gamma}]$
and modify $G^*$ into a plane graph $G^*_{\bar x,\bar y}$
with at most $2\gamma$ terminal faces, as follows.
For every $w_i$, create a single terminal face $f^i_{x_i}$
by ``merging'' faces around $w_i$, starting from $f_{t_{x_i}}$
and going in a clockwise order until $f_{t_{y_i-1}}$ (inclusive).
Then merge similarly the faces from $f_{t_{y_i}}$ and until $f_{t_{x_i-1}}$
into a single terminal face $f^i_{y_i}$.
If $x_i=y_i$, then the two merging operations above are identical,
and thus (as an exception) create only one terminal face denoted $f^i_{x_i}$.
Formally, a merge of two faces is implemented by removing
the edge incident to $w_i$ that goes between the relevant faces.
Observe that removing these edges in $G^*$ can be described in $G$
as contracting the path around the boundary of the face $f_i$
from the terminal $t_{x_i}$ to $t_{y_i-1}$,
and similarly from the terminal $t_{y_i}$ to $t_{x_i-1}$,
see Figure~\ref{FGR: Merge faces in the graph GXY naive bounded gamma}.
It is easy to verify that the modified graph $G^*_{\bar x,\bar y}$ is planar,
and that every elementary cycle $E^*_S$ in $G^*$ with this label
$(\bar x,\bar y)$ is also an elementary cycle in $G^*_{\bar x,\bar y}$
that separates the new terminal faces in a certain way.
Usually, the new terminal faces are separated into
$\{f^i_{x_i}\}_{i=1}^{\gamma}$ and $\{f^i_{y_i}\}_{i=1}^{\gamma}$,
except that when $x_i=y_i$, we have only one new terminal face $f^i_{x_i}$,
which should possibly be included with the $y_i$'s instead of with the $x_i'$.
Since $G^*_{\bar x,\bar y}$ has at most $2\gamma$ terminal faces,
it can have at most $2^{2\gamma}$ elementary cycles (one for each subset).
This shows that for every label $(\bar x,\bar y)$,
there are at most $2^{2\gamma}$ different elementary cycles in $G^*$,
as claimed.

Finally, the number of distinct labels $(\bar x,\bar y)$ is clearly bounded by
$\Pi_{i=1}^{\gamma}k_i^2$ and the above claim applies to each of them. By the inequality of arithmetic and geometric means $ \Pi_{i=1}^{\gamma}k_i^2\leq  (k/\gamma)^{2\gamma}$.
Therefore, the total number of different elementary cycles in $G^*$
is at most $(2k/\gamma)^{2\gamma}$, and Theorem~\ref{THM: naive UB elementary cycles bounded gamma} follows.

\begin{figure}
  \centering
  \includegraphics[angle=0,width=0.5\textwidth]{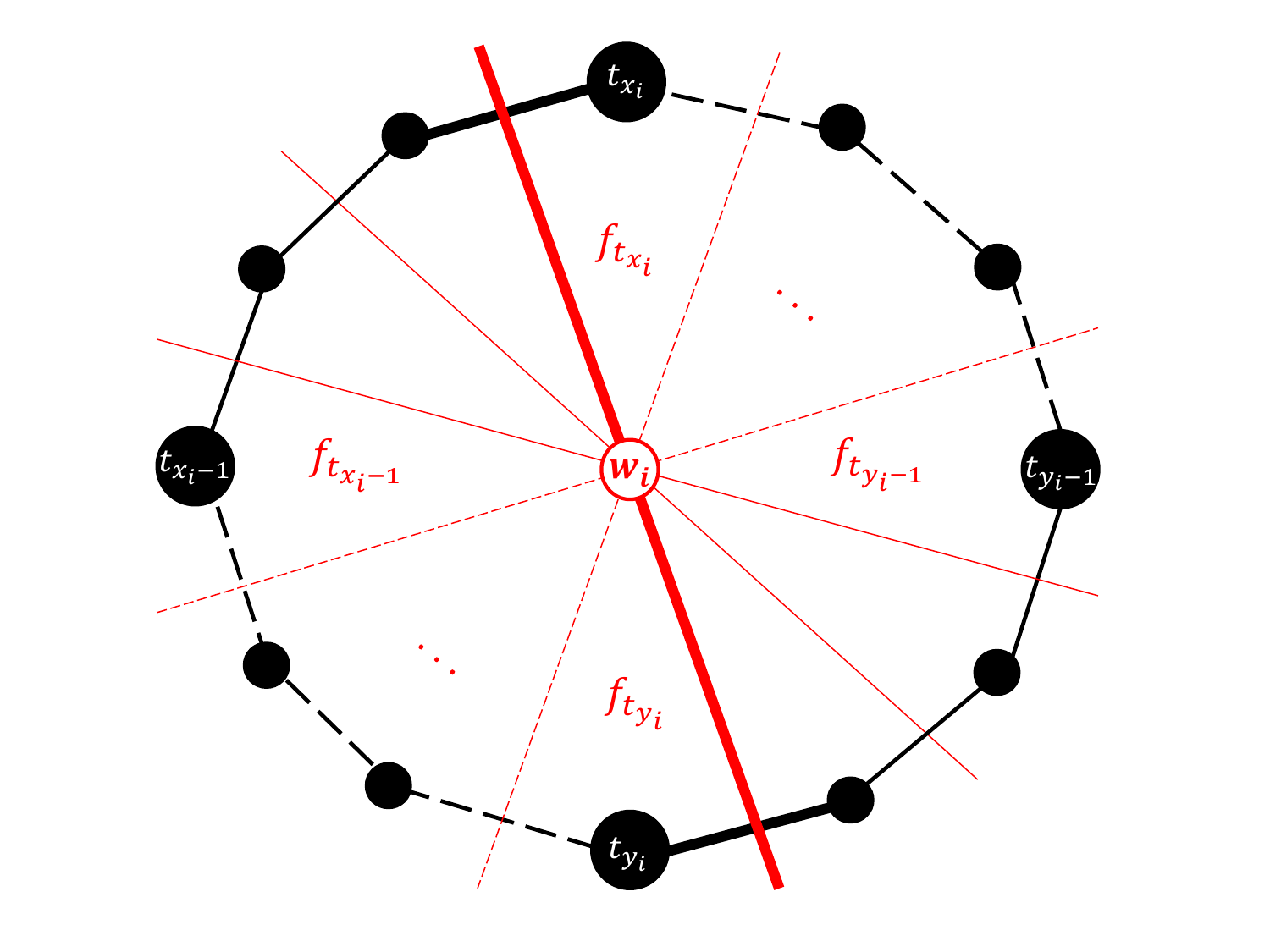}
  \caption{A simple cycle $E^*_S$ separates the faces around $w_i$
    into two subsets.
The primal graph $G$ is shown in black, and its dual $G^*$ in red.
Thicker lines are used for edges of the elementary cutset $E_S$
and of the cycle $E^*_S$.
Dashed lines represent dual edges that are removed by merging faces,
and primal edges that are contracted.
}
  \label{FGR: Merge faces in the graph GXY naive bounded gamma}
  \mbox{}\hrule
\end{figure}

\subsection{Proof of Theorem~\ref{THM: elementary mincuts bounded faces}}
\label{SEC: proof elementary mincuts bounded gamma}

In this section we prove Theorem~\ref{THM: elementary mincuts bounded faces},
which actually decompose the elementary cutsets in a ``bounded manner'' when $\gamma>1$. The idea is to consider the dual graph, which has $\gamma$ special vertices, and elementary cycles. Since every elementary cycle is a simple cycle, it visits each of the $\gamma$ vertices at most once, and thus we can decompose the elementary cycles into paths, such that the two endpoints of every path belong to the $\gamma$ vertices. The challenging part is to count how many distinct paths are there.

We shall use the notation introduced in the beginning of Section~\ref{SEC: naiv UB elementary cuts bounded gamma}. In particular, the graph $G$ has terminals $T=\cupdot_{i=1}^{\gamma} T_i$, where $T_i=\{t^i_{1},\ldots,t^i_{k_i}\}$ are the terminals on the boundary of special face $f_i$, and for simplicity we omit $i$ when it is clear from the context. The dual graph, denoted $G^*$ has terminal faces $\{f_{t^i_j}\}$ and special vertices $W=\{w_1, \ldots, w_{\gamma}\}$. Let $v_{\infty}\in V$ be the vertex whose dual face $f_{v_{\infty}}$ is the outerface of $G^*$.

Informally, the next definition determines whether $f_v$,
the face dual to a vertex $v\in V$,
lies ``inside'' or ``outside'' a circuit $M^*$ in $G^*$.
It works by counting how many times a path from $v$ to $v_\infty$
``crosses'' $M^*$ and evaluating it modulo 2 (i.e., its parity).
The formal definition is more technical because it involves fixing a path,
but the ensuing claim shows the value is actually independent of the path.
Moreover, we need to properly define a ``crossing'' between a path $\Phi$
in $G$ and a circuit in $G^*$;
to this end, we view the path $\Phi$ as a sequence of faces in $G^*$,
that goes from $f_v$ to $f_{v_\infty}$ and at each step ``crosses'' an edge of $G^*$.

\begin{definition}[Parity of a dual face]
Let $f_v$ be the dual face to a vertex $v\in V$,
and fix a simple path in $G$ between $v$ and $v_{\infty}$, denoted $\Phi$.
Let $M^*$ be a circuit in $G^*$,
and observe that its edges $E(M^*)$ form a multiset.
Define the \emph{parity} of $f_v$ with respect to $M^*$ to be
\begin{linenomath}
$$
  \parity(f_v,M^*):=\big(\sum_{e\in E(\Phi)}\Count(e^*,E(M^*))\big) \bmod 2  ,
$$
\end{linenomath}
where $\Count(a,A)$ is the number of times an element $a$ appears
in a multiset $A$.
\end{definition}

The next claim justifies the omission of the path $\Phi$ in the notation $\parity(f_v,M^*)$.

\begin{claim}\label{CLM: same parity for every trail from f_v to outerface}
Fix $v\in V$,
and let $\Phi$ and $\Phi'$ be two paths in $G$ between $v$ and ${v_{\infty}}$.
Then for every circuit $M^*$ in $G^*$,
\begin{linenomath}
$$
  \sum_{e\in E(\Phi)}\Count(e^*,E(M^*))=\sum_{e\in E(\Phi')}\Count(e^*,E(M^*)) \pmod 2.
$$
\end{linenomath}
\end{claim}

\begin{proof}
Fix a vertex $v\in V$ and its dual face $f_v$.
Fix also a circuit $M^*$, and a decomposition of it into simple cycles.
We say that a simple cycle in $G^*$ (like one from the decomposition of $M^*$) \emph{contains} the face $f_v$
if that cycle separates $f_v$ from the outerface $f_{v_\infty}$.
Let $\Phi$ be a path between $v$ and ${v_{\infty}}$.
By the Jordan Curve Theorem, the path's dual edges $\set{e^*: e\in E(\Phi)}$
intersect a simple cycle in $G^*$ an odd number of times
if and only if that simple cycle contains the dual face $f_v$.
By summing this quantity over the simple cycles in the decomposition of $M^*$,
we get that
\begin{linenomath}
$$
  \sum_{e\in E(\Phi)}\Count(e^*,E(M^*)) = 1 \pmod 2
$$
\end{linenomath}
if and only if $f_v$ is contained in an odd number of these simple cycles.
The latter is clearly independent of the path $\Phi$, which proves the claim.
\end{proof}

Given a circuit $M^*$ in $G^*$,
we use the above definition to partition the terminals $T$
into two sets according to their parity, namely,
\begin{linenomath}
\begin{align*}
 T_{odd}(M^*) &\eqdef \{t\in T:\ \parity(f_t,M^*)=1\},
 \\
 T_{even}(M^*) &\eqdef \{t\in T:\ \parity(f_t,M^*)=0\} .
\end{align*}
\end{linenomath}
Given $S\in \T_e(G)$,
recall that $E^*_S$ is the shortest cycle which is $S$-separating in $G^*$ (i.e. it separates between the terminal faces $S^*$ and $\bar{S}^*$).
Since $E^*_S$ is an elementary cycle, it separates the plane into exactly two regions, which implies, \Wlog, $T_{odd}(E^*_S)=S$ and $T_{even}(E^*_S)=\bar S$.
Moreover, $E^*_S$ is a simple cycle and thus goes through every vertex of $W$ at most once.
We decompose $E^*_S$ into $|W\cap V(E^*_S)|$ paths in the obvious way,
where the two endpoints of each path, and only them, are in $W$,
and we let $\Pi_S$ denote this collection of paths in $G^*$.
There are two exceptional cases here;
first, if $|W\cap V(E^*_S)|=1$ then we let $\Pi_S$ contain one path
whose two endpoints are the same vertex (so actually a simple cycle).
second, if $|W\cap V(E^*_S)|=0$ then we let $\Pi_S=\emptyset$
(we will deal with this case separately later).
Now define the set
\begin{linenomath}
$$\Pi\eqdef\bigcup_{S\in \T_e(G)} \Pi_S$$
\end{linenomath}
be the collection of all the paths
that are obtained in this way over all possible $S\in \T_e(G)$.
Notice that if the same path is contained in multiple sets $\Pi_S$,
then it is included in the set $\Pi$ only once
(in fact, this ``overlap'' is what we are trying to leverage).

Now give to each path $P \in \Pi$ a label that consists of three parts:
(1) the two endpoints of $P$, say $w_i,w_j\in W$;
(2) the two successive terminals on each of the faces $f_i$ and $f_j$,
which describe where the path $P$ enters vertices $w_i$ and $w_j$,
say between $t^i_{x-1},t^i_{x}$ and between $t^j_y,t^j_{y+1}$; and
(3) the set $T_{odd}(P\cup \Pi_{ij})$,
where $\Pi_{ij}$ is the shortest path (or any other fixed path)
that agrees with parts (1) and (2) of the label and does not go through $W$,
i.e., the shortest path between $w_i$ and $w_j$
that enters them between $t^i_{x-1},t^i_{x}$ and $t^j_y,t^j_{y+1}$
and does not go through any other vertex in $W$.
This includes the exceptional case $i=j$,
in which $P$ is actually a simple cycle.

We proceed to show that each label is given to at most one path in $\Pi$
(which will be used to bound $\card{\Pi}$).
Assume toward contradiction that two different paths $P,P'\in \Pi$ get the same label, and suppose $c(P')<c(P)$.
Suppose $P$ is the path between $w_i$ to $w_j$ in $E^*_S$ for $S\in T_e(G)$,
and $P'$ is the path between the same endpoints (because of the same label)
in $E^*_{S'}$ for another $S'\in T_e(G)$.
By construction, the paths $P$ and $P'$ are simple,
because $E^*_S$ and $E^*_{S'}$ are elementary cycles,
and only their endpoint vertices are from $W$.

The key to arriving at a contradiction is the next lemma.
In these proofs, a path $P$ is viewed as a multiset of edges $E(P)$,
and the union and subtraction operations are applied to multisets.
In particular, the union of two paths with the same endpoints gives a circuit.

\begin{lemma}\label{LMA: Switch subpaths in E^*_S}
The circuit $(E^*_S\setminus P)\cup P'$ is $S$-separating.
\end{lemma}

To prove this lemma, we will need the following two claims.

\begin{claim}\label{CLM: parity of A B C}
Let $A,B$ and $C$ be (the edge sets of) simple paths in $G^*$
between the same $w_i,w_j\in W$. Then
\begin{linenomath}
$$
  \forall t\in T,\qquad
  \parity(f_t, A\cup C)=\parity(f_t, A\cup B) + \parity(f_t, B\cup C) \pmod 2.
$$
\end{linenomath}
\end{claim}

\begin{proof}
Fix $t\in T$ and a path $\Phi$ between $t$ and ${v_{\infty}}$.
Since $A,B$ and $C$ are simple paths,
\begin{linenomath}
\begin{align*}
  \sum_{e\in E(\Phi)}\Count(e^*,A\cup C)
  &= |E^*(\Phi)\cap A|+|E^*(\Phi)\cap C|,
  \\
  \sum_{e\in E(\Phi)}\Count(e^*,A\cup B)
  &=|E^*(\Phi)\cap A|+|E^*(\Phi)\cap B|,
  \\
  \sum_{e\in E(\Phi)}\Count(e^*,B\cup C)
  &=|E^*(\Phi)\cap B|+|E^*(\Phi)\cap C|.
\end{align*}
\end{linenomath}
Summing the three equations above modulo $2$ yields
\begin{linenomath}
\begin{align*}
  \parity(f_t, A\cup C)
  + \parity(f_t, A\cup B) + \parity(f_t, B\cup C)
  = 0  \pmod 2,
\end{align*}
\end{linenomath}
which proves the claim.
\end{proof}

\begin{claim}\label{CLM: odd terminals switch parity in A B C}
  Let $\bigtriangleup$ be the symmetric difference between two sets.
  For every 3 simple paths $A,B$ and $C$ between $w_i,w_j\in W$,
  \begin{linenomath}
  $$T_{odd}(A\cup C)=T_{odd}(A\cup B)\bigtriangleup T_{odd}(B\cup C).$$
  \end{linenomath}
\end{claim}

\begin{proof}
Observe that $T_{odd}(A\cup B)\bigtriangleup T_{odd}(B\cup C)$
contains all $t\in T$ for which exactly one of
$\parity(f_t, A\cup B)$ and $\parity(f_t, B\cup C)$ is equal to $1$,
which by Claim~\ref{CLM: parity of A B C} is equivalent to having
$\parity(f_t, A\cup C)=1$.
\end{proof}

\begin{proof}[Proof of Lemma \ref{LMA: Switch subpaths in E^*_S}]
To set up some notation,
let $Q\eqdef E^*_S\setminus P$ be a simple path between $w_i$ and $w_j$.
Since $E^*_S$ is a simple cycle that contains $P$, we can write $E^*_S=Q\cup P$.

The idea is to swap the path $P$ in $E^*_S$ with the other path $P'$,
which for sake of analysis is implemented in two steps.
The first step replace $P$ (in $E^*_S$) with $\Pi_{ij}$,
which gives the circuit $(E^*_S\setminus P)\cup \Pi_{ij} = Q\cup \Pi_{ij}$.
The second step replaces $\Pi_{ij}$ with $P'$,
which results with the circuit $Q\cup P'=(E^*_S\setminus P)\cup P'$.
Now apply Claim~\ref{CLM: odd terminals switch parity in A B C} twice,
once to the simple paths $A=Q$, $B=P$ and $C=\Pi_{ij}$,
and once to the simple paths $A=Q$, $B=\Pi_{ij}$ and $C=P'$,
we get that
\begin{linenomath}
\begin{align*}
  T_{odd}(Q\cup \Pi_{ij})
  &= T_{odd}(E^*_S)\bigtriangleup T_{odd}(P\cup \Pi_{ij}) ,
  \\
  T_{odd}(Q\cup P')
  &= T_{odd}(Q\cup \Pi_{ij}) \bigtriangleup T_{odd}(\Pi_{ij}\cup P').
\end{align*}
\end{linenomath}
By plugging the first equality above into the second one,
and observing that $T_{odd}(\Pi_{ij}\cup P) = T_{odd}(\Pi_{ij}\cup P')$
because $P$ and $P'$ have the same label,
we obtain that
\begin{linenomath}
\begin{equation}\label{EQ: same odd terminal faces}
  T_{odd}(Q\cup P')
    =T_{odd}(E^*_S) .
\end{equation}
\end{linenomath}

Finally, it is easy to verify that the circuit $Q\cup P'$
must separate between $T_{odd}(Q\cup P')$ and $T_{even}(Q\cup P')$.
Using~\eqref{EQ: same odd terminal faces}
and the fact that $E^*_S$ is an elementary cycle,
we know that $T_{odd}(Q\cup P')=T_{odd}(E^*_S)=S$,
and thus $T_{even}(Q\cup P')=T\setminus S$.
It follows that $Q\cup P'$ is $S$-separating, as required.
\end{proof}

Lemma~\ref{LMA: Switch subpaths in E^*_S} shows that
the circuit $(E^*_S\setminus P)\cup P'$ is $S$-separating,
while also having lower cost than $E^*_S$.
This contradicts the minimality of $E^*_S$,
and shows that the paths in $\Pi$ have distinct labels.
Thus, $\card{\Pi}$ is at most the number of distinct labels,
and we will bound the latter using the following claim.

\begin{claim}\label{CLM: all terminals of same w have same parity}
Let $P\in \Pi$ be a path between $w_i$ and $w_j$, and let $r\in [\gamma]$.
Then
\begin{linenomath}
$$
  \forall t,t'\in T_r, \qquad
  \parity(f_t, P\cup \Pi_{ij})=\parity(f_{t'}, P\cup \Pi_{ij}),
$$
\end{linenomath}
where $\Pi_{ij}$ is the shortest path with the same parts (1) and (2) of the label as $P$, and does not go through any other vertices of $W$.
\end{claim}

\begin{proof}
Since $t,t'\in T_r$,
their dual faces $f_t$ and $f_{t'}$ share $w_r$ on their boundary.
$P$ and $\Pi_{ij}$ are simple paths in $G^*$ with the same endpoints,
and thus $P\cup \Pi_{ij}$ is a circuit in $G^*$,
which by construction does not go through any vertex $w_r$ with $r\neq i,j$. Fix a path $\Phi$ in $G$ between $t$ and $v_\infty$.
We can extend it into a path $\Phi'$ between $t'$ and $v_\infty$, by taking a path $A_{t't}$ in $G$ that goes around the face $f_r$
between $t'$ and $t$ (both are on the face $f_r$, because $t,t'\in T_r$),
and letting $\Phi' \eqdef A_{t't}\cup\Phi$.

Since $P$ and $\Pi_{ij}$ agree on the same parts (1) and (2) of the label, then $P\cup \Pi_{ij}$ have exactly two edges between some two successive terminals on each of the faces $f_i$ and $f_j$. Thus, if $r\neq i,j$ then $\card{A_{t't}\cap (P\cup \Pi_{ij})}=0$. If $r=i$ or $r=j$ but $i\neq j$ then $\card{A_{t't}\cap (P\cup \Pi_{ij})}$ is either 0 or 2. And if $r=i=j$ then $\card{A_{t't}\cap (P\cup \Pi_{ij})}$ is either 0, 2 or 4. Therefore, if we examine the parities of $f_t$ and $f_{t'}$ with respect to $P\cup \Pi_{ij}$ using the paths $\Phi$ and $\Phi'=A_{t't}\cup\Phi$, respectively,
we conclude that these parities are equal, as required.
\end{proof}

We can now bound the number of possible labels of a path $P\in\Pi$.
There are $\gamma^2$ possibilities for part 1 of the label,
i.e., the endpoints $w_i,w_j\in W$ of $P$ (note that we may have $i=j$).
Given this data, there are $k_ik_j$ possibilities for part 2,
i.e., between which two terminals the path $P$ exits $w_i$ and enters $w_j$.
Furthermore, the number of possibilities for part 3 is the number of different subsets $T_{odd}(P\cup\Pi_{ij})$. By Claim~\ref{CLM: all terminals of same w have same parity}
for every $r\in[\gamma]$ either $T_r\subseteq T_{odd}(P\cup\Pi_{ij})$ or $T_r\cap T_{odd}(P\cup\Pi_{ij})=\emptyset$. Thus, the number of different subsets
$T_{odd}(P\cup\Pi_{ij})$ %
is the number of different subsets of $\{T_1,\ldots,T_{\gamma}\}$, which is at most $2^{\gamma}$.
Altogether we get that there are at most $2^{\gamma}\sum_{i,j=1}^{\gamma}k_i\cdot k_j$ different labels.

Finally, there are also cycles $E^*_S$ for $S\in \T_e(G)$ that do not go through any vertices of $W$, i.e. $W\cap V(E^*_S)=\emptyset$. Thus, they are not include in $\Pi$, so we count them now separately. Recall that \Wlog $T_{odd}(E^*_S)=S$, i.e every such cycle $E^*_S$ is identified uniquely by a different subset $T_{odd}(\cdot)$. Since by Claim~\ref{CLM: all terminals of same w have same parity} there are at most $2^{\gamma}$ such subsets, we get that there are at most $2^{\gamma}$ such cycles. Adding them to our calculation, and Theorem~\ref{THM: elementary mincuts bounded faces} follows.

\begin{figure}[t]
        \centering
        \includegraphics[angle=0,width=0.8\textwidth]{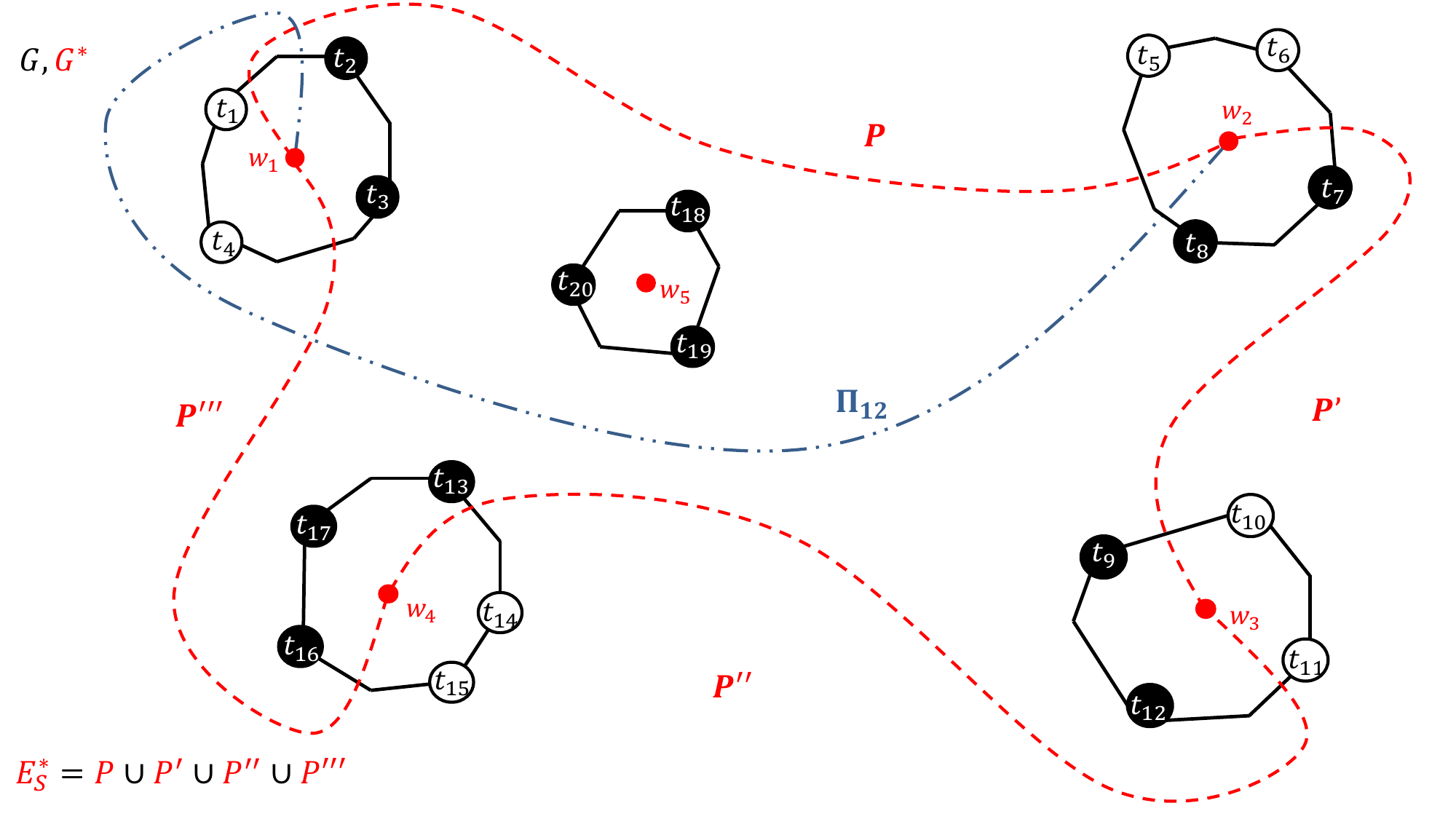}
        \caption{A planar 20-terminal network with $\gamma=5$.
          Let $S\subset T$ be all the black terminals,
          then $E^*_S$ (red dashed line)
          is split into $4$ paths $P, P', P'', P'''$.
          The label of $P$, for example, is
          (1) $w_1$ and $w_2$; (2) $t_1,t_2$ and $t_5,t_8$; (3) $T_{odd}(P\cup \Pi_{12})=\{t_1,t_2,t_3,t_4,t_{18},t_{19},t_{20}\}$,
          and is computed using $\Pi_{12}$ (blue dashed line).
        }
        \label{FGR: PlanarGraphBoundedFaces}
        \mbox{}\hrule
\end{figure}

\subsection{Flow Sparsifiers}
\label{sec:flow}

Okamura and Seymour \cite{OS81} proved that
in every planar network with $\gamma(G)=1$, the flow-cut gap is $1$
(as usual, flow refers here to multicommodity flow between terminals).
It follows immediately, see e.g.~\cite{AGK14}, that for such a graph $G$,
every $(q,s)$-cut-sparsifier is itself also a $(q,s)$-flow-sparsifier of $G$.
Thus, Corollary~\ref{THM: mimick network size k^4} implies the following.

\begin{corollary}\label{COR: flow sparsifier size k^4}
Every planar $k$-terminal network $G$ with $\gamma=1$ admits
a minor $(1,O(k^4))$-flow-sparsifier.
\end{corollary}

Chekuri, Shepherd, and Weibel~\cite[Theorem 4.13]{CSW13} proved that
in every planar network $G$, the flow-cut gap is at most $3\gamma(G)$,
and thus Corollary~\ref{THM: UB mimicking bounded faces} implies the following.

\begin{corollary}
Every planar $k$-terminal network $G$ with $\gamma=\gamma(G)$
admits a minor $(3\gamma, O(\gamma 2^{2\gamma}k^4))$-flow-sparsifier.
\end{corollary}

\section{Terminal-Cuts Scheme}
\label{SEC: Terminal Cut scheme}
In this section we present applications
of our results in Section~\ref{SEC: Elementary cuts General graphs}
to data structures that store all the minimum terminal cuts in a graph $G$.
As our focus is on the data structure's memory requirement,
we do not discuss its query time.
We start with a formal definition of such a data structure,
and then provide our bounds
of $\tilde{O}(|\T_e|)$ bits for general graphs
(Theorem~\ref{THM: UB data Structure}),
and $\tilde{O}(2^{\gamma}k^2)$ bits for planar graphs
(Corollaries~\ref{COR: UB k^2 data Structure}
and~\ref{THM: UB data structure gamma(G)}). In comparison, a trivial data structure for general graphs uses $\tilde{O}(2^k)$ bits, by storing the cost of all the terminal mincuts explicitly.

\begin{definition}\label{DEF:TCscheme}
A \emph{terminal-cuts scheme (\TCs)} is a data structure
that uses a storage (memory) $M$ to support the following two operations on
a $k$-terminal network $G=(V,E,T,c)$,
where $n=|V|$ and $c:E \to \{1,\ldots,n^{O(1)}\}$.
\begin{enumerate}
\item Preprocessing, denoted $P$, which gets as input the network $G$ and builds $M$.
\item Query, denoted $R$, which gets as input a subset of terminals $S$,
and uses $M$ (without access to $G$) to output the cost of the minimum cutset $E_S$.
\end{enumerate}
\end{definition}

We usually assume a machine word size of $O(\log n)$ bits,
because even if $G$ has only unit-weight edges,
the cost of a cut might be $O(n^2)$, which is not bounded in terms of $k$.

\begin{theorem}%
\label{THM: UB data Structure}
Every $k$-terminal network $G=(V,E,T,c)$
admits a \TCs with storage size of $O(|\T_e(G)|(k+\log n))$ bits,
where $\T_e(G)$ is the set of elementary cutsets in $G$.
\end{theorem}

\begin{proof}%
We construct a \TCs as follows.
In the preprocessing stage, given $G$,
the \TCs stores $\langle S, c(E_S) \rangle$ for every $S\in \T_e(G)$,
where $S$ is written using $k$ bits.
The cost of every cutset is at most $|E|\cdot n^{O(1)}=\poly(n)$,
and thus the storage size of the \TCs is $O(|\T_e(G)|(k+\log n))$ bits,
as required.
Now given a subset $S\subset T$, the query operation $R(S;P(G))$
outputs
\begin{linenomath}
\begin{equation} \label{EQ: TCs Query}
  \min \Big\{\sum_{S'\in \varphi } \mincut_G(S') :\
  \varphi \subseteq \T_e(G)
  \text{ s.t. $\mycup_{S'\in \varphi}\mincutset_G(S')$ is
    $S$-separating in $G$} \Big\} .
\end{equation}
\end{linenomath}
Since for every $\varphi\subseteq 2^T$,
the cutset $\cup_{S'\in \varphi}\mincutset_G(S')$ is $S$-separating in $G$
if and only if
$|\1_S(t_i)-\1_S(t_j)|\leq\sum_{S'\in \varphi}|\1_{S'}(t_i)-\1_{S'}(t_j)|$
for all $i,j\in[k]$,
the calculation in \eqref{EQ: TCs Query} can be done with no access to $G$.
Clearly, $\mincut_G(S)\leq R(S;P(G))$.
By Theorem \ref{THM: E_S disjoint union elementary cuts},
there is $\varphi\subseteq {\T_e(G)}$ such that
$\mincutset_G(S)=\cupdot_{S'\in \varphi}\mincutset_G(S')$
and $\mincut_G(S)=\sum_{S'\in \varphi} \mincut_G(S')$.
Thus, $R(S';P(G))=\mincut_G(S')$. %
\end{proof}

\begin{corollary} \label{COR: UB k^2 data Structure}
Every planar $k$-terminal network $G$ with $\gamma=1$
admits a \TCs with storage size of $O(k^2 \log n)$ bits,
i.e., $O(k^2)$ words.
\end{corollary}

\begin{proof}
If $G$ is a planar $k$-terminal network with $\gamma=1$,
then by Theorem~\ref{THM: at most k^2 elementary min cuts} every $S\in \T_e(G)$ is equal to $\{{t_i},{t_{i+1}}, \ldots, {t_j}\}$ for some $i,j\in [k]$ and $|\T_e(G)|=\binom{k}{2}$ (recall that all the terminals $t_1,\ldots,t_k$ are on the outerfaces of $G$ in order).
Thus, we can specify $S$ via these two indices $i$ and $j$,
using only $O(\log k) \leq O(\log n)$ bits (instead of $k$).
The storage bound follows.
\end{proof}

\begin{theorem} \label{THM: UB data structure gamma(G)}
Every planar $k$-terminal network $G$ with $\gamma=\gamma(G)$ admits
a \TCs with storage size of
$O\Big(2^{\gamma}\big( 1+\sum_{i,j=1}^{\gamma} k_i\cdot k_j \big)\big(\gamma+\log n\big)\Big)
  \leq O(2^{\gamma}k^2(\gamma+\log n))$
bits.
\end{theorem}

\begin{proof}[Proof sketch]
If $G$ is a planar $k$-terminal network with bounded $\gamma$, then Theorem \ref{THM: elementary mincuts bounded faces} characterize $2^{\gamma}k^2$ special subsets of edges together with some small addition information for each such subset that denote by \emph{label}. It further prove that all the elementary cuts can be restored using only the special subsets and their labels. As each label can be stored by at most $O(\gamma)$ bits, the storage bound follows.
\end{proof}

\section{Cut-Sparsifier vs. DAM in planar networks}
\label{app: Duality cuts and distances}
In this section we prove the duality between cuts and distances in planar graphs with all terminals on the outerface. Although the duality between shortest cycles and minimum cuts in planar graphs is known, the main difficulty is to transform all the shortest cycles into shortest paths without blowing up the number of terminals in the graph. We prove this duality using the following two theorems, and applications of them can be found in Section \ref{SEC: Reduction Application}.

\begin{theorem}\label{THM: reduction from outer face cuts to DAM}
    Let $G=(V,E,T,c)$ be a planar $k$-terminal network with all its terminals $T$ on the outerface. One can construct a planar $k$-terminal network $G'=(V',E',T',c')$ with all its terminals $T'$ on the outerface, such that if $G'$ admits a $(q, s)$-DAM then $G$ admits a minor $(q,O(s))$-cut-sparsifier.
\end{theorem}

\begin{theorem}\label{THM: reduction from outer face DAM to outer face cuts}
Let $G=(V,E,T,c)$ be a planar $k$-terminal network with all its terminals $T$ on the outerface. One can construct a planar $k$-terminal network $G'=(V',E',T',c')$ with all its terminals $T'$ on the outerface, such that if $G'$ admits minor $(q,s)$-cut-sparsifier then $G$ admits a $(q,O(s))$-DAM.
\end{theorem}

\subsection{Proof of Theorem~\ref{THM: reduction from outer face cuts to DAM}}
\label{app: Reduction CAM to DAM}

\paragraph{Construction of the Reduction.}
The idea is to first use the duality of planar graphs in order to convert every minimum terminal cut into a shortest cycle, and then ``open'' every shortest cycle into a shortest path between two terminals, which in turn are preserved by a $(q,s)$-DAM. More formally, given a plane $k$-terminal network $G=(V,E,F,T,c)$ with all its terminals $T=\{t_1,\ldots,t_k\}$ on the outerface in a clockwise order, we firstly construct its dual graph $G_1$ where the boundaries of all its $k$ terminal faces $T(G_1)=\{f_{t_{1}},\ldots,f_{t_{k}}\}$ share the same vertex $v_{f_\infty}$, and secondly we construct $G_2$ by the graph $G_1$ where the vertex $v_{f_\infty}$ is split into $k$ different vertices $v^{i,i+1}_{f_\infty}$, and every edge $(v_{f_\infty},v^*)$ that embedded between (or on) the two terminal faces $f_{t_{i}}$ and $f_{t_{i+1}}$ in $G_1$ correspond to a new edge $(v^{i,i+1}_{f_\infty},v^*)$ in $G_2$ with the same length. See Figure \ref{FGR: Reduction Outer Mimicking To Outer DPM} from left to right for illustration, and see Appendix~\ref{app: Planar Duality} for basic notions of planar duality.
In the following, $f+f'$ denotes a new face that is the union of two faces $f$ and $f'$.

\begin{linenomath}
    \begin{align*}
      V(G_2) := & \big(V(G_1)\setminus \{v_{f_\infty}\}\big) \cup \{v^{1,2}_{f_\infty},\ldots, v^{k-1,k}_{f_\infty},v^{k,1}_{f_\infty}\}
      \\
      E(G_2) := & \big(E(G_1) \setminus \{(v_{f_\infty},v^*)\ :\ v^*\in V(G_1)\}\big)
      \\
      & \cup \{(v^{i,i+1}_{f\infty},v^*)\ :\ i\in[k], (v_{f_\infty},v^*)\in E(G_1), v^*\text{ between } f_{t_{i}}, f_{t_{i+1}}\}\text{\footnotemark}
      \\
      F(G_2) := & \big(F(G_1)\setminus \{f_\infty,f_{t_1},\ldots, f_{t_k}\}\big) \cup \{f_\infty+ f_{t_1} + \ldots + f_{t_k}\}
      \\
      T(G_2) := & \{v^{1,2}_{f_\infty},\ldots, v^{k-1,k}_{f_\infty},v^{k,1}_{f_\infty}\}
    \end{align*}
\end{linenomath}

\footnotetext{We allow wraparound, i.e., $v^{k,k+1}=v^{k,1}$.}

    Let $H_2$ be an $(q, s)$-DAM of $G_2$. Since it is a minor of $G_2$, both are planar $k$-terminals network such that all their terminals are on their outerface in the same clockwise order. Hence, we can use $H_2$ and the same reduction above, but in reverse operations, in order to construct a $(q,O(s))$-cut-sparsifier $H$ for $G$. First, we ``close'' all the shortest paths in $H_2$ into cycles by merging its $k$ terminals $v^{i-1,i}_{f_\infty}$ into one vertex called $v_{f_\infty}$, and denote this new graph by $H_1$. Note that $H_1$ has $k$ new faces $f_{t_1},\ldots, f_{t_k}$, where each face $f_{t_i}$ was created by uniting the two terminals $v^{i-1,i}_{f_\infty}, v^{i,i+1}_{f_\infty}$ of $H_2$. These $k$ new faces of $H_1$ will be its $k$ terminal faces. Secondly, we argue that the dual graph of $H_1$ is our requested cut-sparsifier of $G$, which we denote by $H$. See Figure \ref{FGR: Reduction Outer Mimicking To Outer DPM} from right to left for illustration.

\begin{figure}
  \centering
  \includegraphics[angle=0,width=0.9\textwidth] {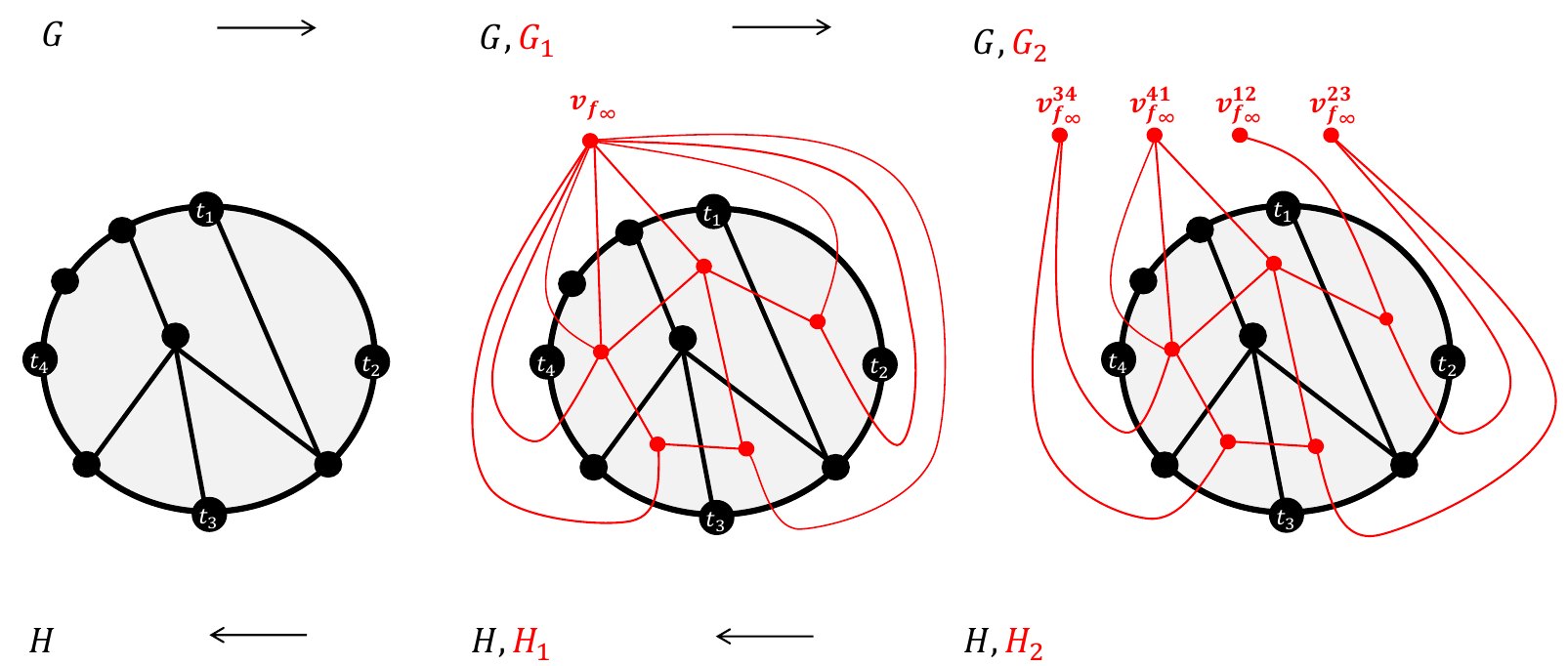}
  \caption{The first graph (in black) is the original graph $G$. The second is its dual graph $G_1$ colored in red. The third graph $H$ and its new terminals colored in red. We get that graph by ``splitting'' the outerface vertex $v_{f_\infty}$ to $k$ new vertices, which are the new terminals.}
  \label{FGR: Reduction Outer Mimicking To Outer DPM}
  \mbox{}\hrule
\end{figure}

\paragraph{Analysis of the Reduction.}
The key element of the reduction's proof is the duality between every shortest cycle in $G_1$ to a shortest path in $G_2$, which we formally stated in the following lemma. Given $G$ and its dual graph $G_1=G^*$ as stated above, for every subsets of terminals $S\subset T(G)$ we denote by $S^*\subset T(G^*)$ the corresponding set of terminal faces, i.e. $S^*=\{f_{t_i}:\ \forall i\in[k] \text{ s.t. } t_i\in S\}$.

\begin{lemma}\label{LMA: equivalence minimum circuits and shortest paths}
Every shortest circuit that separates between the terminal faces $S^*_{i,j}$ and $\bar S^*_{i,j}$ in $G_1$, corresponds to a shortest path between the two terminals $v^{i-1,i}_{f_\infty}$ and $v^{j,j+1}_{f_\infty}$ in $G_2$, and vise versa.
\end{lemma}

\begin{linenomath}
\vspace{-.1in}
\begin{align*}
        V(H_1) & :=\big(V(H_2)\setminus \{v^{1,2}_{f_\infty},\ldots, v^{k-1,k}_{f_\infty},v^{k,1}_{f_\infty}\}\big)\cup\{v_{f_\infty}\} \\
        E(H_1) & :=\big(E(H_2) \setminus\{(v^{i,i+1}_{f\infty},v^*)\ :\ i\in[k]\}\big) \cup \{(v_{f_\infty},v)\ :\ (v^{i,i+1}_{f\infty},v^*)\in E(H_2)\} \\
        F(H_1) & :=\big(F(H_2)\setminus \{f_\infty+ f_{t_1} + \ldots + f_{t_k}\}\big)\cup \{f_\infty,f_{t_1},\ldots, f_{t_k}\} \\
        T(H_1) & :=\{f_{t_1},\ldots, f_{t_k}\}
\end{align*}
\vspace{-.1in}
\end{linenomath}

    \begin{proof}%

        \textbf{First direction - circuits to distances.}

        Let $\mathcal{C}$ be a minimum circuit that separates between the terminal faces $S^*_{ij}$ and ${\bar S^*}_{ij}$ in $G_1$ (assume \Wlog $i\leq j$). By Theorem \ref{THM: E_S disjoint union elementary cuts} that circuit is a union of a disjoint shortest $l$ cycles for some $l\geq 1$. We prove that this circuit corresponds to a simple path in $G_2$ between the terminals $v^{i-1,i}_{f_\infty}$ and $v^{j,j+1}_{f_\infty}$ with the same weight using an induction on $l$.

        Induction base: $l=1$. The circuit $\mathcal{C}$ contains exactly one simple cycle $C$ that separates between the terminal faces $S^*_{ij}$ and ${\bar S^*}_{ij}=S^*_{(j+1)(i-1)}$ in $G_1$. So the vertex $v_{f_\infty}$ appear in $C$ exactly once, i.e. $C=\langle v_{f_\infty},v_1,v_2,\ldots, v_x,v_{f_\infty}\rangle$. According to our construction, the graph $G_2$ contains the same vertices and edges as $G_1$, except of the vertex $v_{f_\infty}$ and all the edges incident to it. Therefore, $\langle v_1,v_2,\ldots, v_x\rangle$ is a simple path in $G_2$. Moreover, since \Wlog the vertex $v_1$ embedded between the terminal faces $f_{t_{i-1}},f_{t_i}$ and the vertex $v_x$ embedded between the terminal faces $f_{t_{j}},f_{t_{j+1}}$, we get that $(v^{i-1,i}_{f_\infty},v_1),\ (v_x,v^{j,j+1}_{f_\infty})\in E(G_2)$. Thus, $\langle v^{i-1,i}_{f_\infty},v_1,v_2,\ldots, v_x,v^{j,j+1}_{f_\infty}\rangle$ is a simple path in $G_2$ with the same weight as $C$.

        Induction step: assume that if $\mathcal{C}$ has $l'<l$ cycles, then it corresponds to a simple path in $G_2$ between the terminals $v^{i-1,i}_{f_\infty}$ and $v^{j,j+1}_{f_\infty}$ with the same weight, and prove it for $l$. There are two cases:

        \begin{itemize}

          \item If neither of the cycles in the circuit is nested. Then \Wlog all the cycles $C\in \mathcal{C}$ bound terminal faces of $S^*_{ij}$. Let $C\in \mathcal{C}$ be the cycle that bound the terminal faces $f_{t_i}, \ldots, f_{t_x}$ were $i<x<j$. Thus $\{C\}$ is a simple circuit that separates between the terminal faces $S^*_{ix}$ to $\bar S^*_{ix}=S^*_{(x+1)(i-1)}$, and $\mathcal{C}\setminus \{C\}$ is a simple circuit that separates between the terminal faces $S^*_{(x+1)j}$ to $\bar S^*_{(x+1)j}=S^*_{(j+1)x}$ in $G_1$. By the inductive assumption these two circuits correspond to two simple paths in $G_2$ with the same weights. The first path is between the two terminals $v^{i-1,i}_{f_\infty}$ and $v^{x,x+1}_{f_\infty}$, and the second is between the two terminals $v^{x,x+1}_{f_\infty}$ and $v^{j,j+1}_{f_\infty}$, which form a simple path from $v^{i-1,i}_{f_\infty}$ to $v^{j,j+1}_{f_\infty}$ in $G_2$ with the same weight as $\mathcal{C}$.

          \item There are nested cycles in the circuit. Let $C\in \mathcal{C}$ be a simple cycle that separates between $S^*_{xy}$ and $\bar S^*_{xy}$ in $G_2$, and contains at least one cycle of $\mathcal{C}\setminus \{C\}$. If $i< x\leq y<j$ or $x<i\leq j <y$ then $\mathcal{C}\setminus \{C\}$ separates between $S^*_{ij}$ to $\bar{S^*_{ij}}$ in contradiction to the minimality of $\mathcal{C}$. Therefore either $i= x\leq j<y$ or $i< x\leq j=y$. Assume \Wlog that the first case holds, i.e. $C$ is a minimum circuit that separates between $S^*_{iy}$ and $\bar S^*_{iy}$, and $\mathcal{C}\setminus \{C\}$ is a minimum circuit that separates between the terminal faces $S^*_{( y+1)j}$ to $\bar S^*_{( y+1)j}$ in $G_1$. By the inductive assumption these two circuits correspond to two simple paths in $G_2$ with the same weights. The first simple path is between the two terminals $v^{i-1,i}_{f_\infty}$ and $v^{ y, y+1}_{f_\infty}$, and the second simple path is between the two terminals $v^{ y, y+1}_{f_\infty}$ and $v^{j,j+1}_{f_\infty}$. Uniting these two paths forms a simple path between $v^{i-1,i}_{f_\infty}$ to $v^{j,j+1}_{f_\infty}$ in $G_2$ with the same weight as $\mathcal{C}$ as we required.
        \end{itemize}

        \textbf{Second direction - distances to cuts.} Let $P$ be a shortest path between the terminals $v^{i-1,i}_{f_\infty}$ and $v^{j,j+1}_{f_\infty}$ in $G_1$ (assume $i\leq j$), and let $l$ be the number of terminals in that path (including the two terminals in its endpoints). It is easy to verify that replacing each terminal $v^{x,x+1}_{f_\infty}$ in $P$ with the vertex $v_{f_\infty}$ transform it to a circuit in $G_1$ with $l-1$ disjoint simple cycles and with the same weight of $P$. We prove that this circuit separates between the terminal faces $S^*_{ij}=\{f_{t_i}, \ldots, f_{t_j}\}$ and $\bar S^*_{ij}$ in $G_1$ by an induction on $l$.

        Induction base: $l=2$, i.e. the only terminals on the path $P$ are those on the endpoints.
        Thus, all the inner vertices on that path are non terminal vertices, i.e. $P=\langle v^{i-1,i}_{f_\infty},v_1,v_2,\ldots, v_x,v^{j,j+1}_{f_\infty}\rangle$. Substitute the terminals $v^{i-1,i}_{f_\infty}$ and $v^{j,j+1}_{f_\infty}$ of $G_2$ with the vertex $v_{f_\infty}$ of $G_1$ and get $C=\langle v_{f_\infty},v_1,v_2,\ldots, v_x,v_{f_\infty}\rangle$. According to our construction, $\langle v_1,v_2,\ldots, v_x\rangle$ is a simple path in $G_1$, and $(v^{i-1,i}_{f_\infty},v_1),(v_l,v^{j,j+1}_{f_\infty})\in E(G_2)$ if and only if $(v_{f_\infty},v_1),(v_x,v_{f_\infty})\in E(G_1)$. Therefore, $C$ is a simple cycle in $G_1$ and the two edges that incident to the vertex $v_{f_\infty}$ are embedded between the terminal faces $f_{t_{i-1}}$ to $f_{t_i}$ and $f_{t_{j}}$ to $f_{t_{j+1}}$ in $G_1$. Thus, $C$ separates between $S^*_{ij}$ to $\bar S^*_{ij}$, and has the same weight as $P$.

        Induction step: assume that if $P$ has $l'<l$ inner terminals then it corresponds to a simple circuit with $l'$ cycles that separates between the terminal faces $S^*_{ij}$ and $\bar S^*_{ij}$ in $G_1$, and prove it for $l'=l$. Let $v^{x,x+1}_{f_\infty}$ be some inner terminal in the path $P$ that brake it into two simple sub-paths $P_1$ and $P_2$, i.e. $P_1$ is a simple path between $v^{i-1,i}_{f_\infty}$ to $v^{x,x+1}_{f_\infty}$ and $P_2$ is a simple path between $v^{x,x+1}_{f_\infty}$ to $v^{j,j+1}_{f_\infty}$ in $G_2$. Since both of these paths have less than $l$ terminals we can use the inductive assumption and get that $P_1$ corresponds to a circuit $\mathcal{C}_1$ in $G_1$ with the same weight that separates between the terminals $S^*_{ix}$and $\bar S^*_{ix}$ and $\bar S^*_{ix}$, and $P_2$ corresponds to a circuit $\mathcal{C}_2$ in $G_1$ with the same weight that separates between the terminals $S^*_{(x+1)j}$ and $\bar S^*_{(x+1)j}$. If $i\leq x \leq j$, then $S^*_{ij}=S^*_{ix}\cup S^*_{(x+1)j}$. And if $i\leq j<x$ (symmetric to the case were $x < i \leq j$), then $S^*_{(x+1)j}=S^*_{(j+1)x}$ and so $S^*_{ij}=S^*_{ix}\setminus S^*_{(j+1)x}$. In both cases we get that $\mathcal{C}_1 \cup \mathcal{C}_2$ is a simple circuit in $G_1$ with the same weight as $P$ that separates between the terminal faces $S^*_{ij}$ and ${\bar S^*}_{ij}$ in $G_1$, and the Lemma follows.
    \end{proof}

    \begin{lemma}\label{LMA: reduction H is alpha cut sparsifier of G}
      The elementary cuts $\T_e(G)$ and $\T_e(H)$ are equal, and $\mincut_G(S) \leq \mincut_H(S) \leq q \cdot \mincut_G(S)$ for every $S\in \T_e(G)$.
    \end{lemma}

    \begin{proof}
        Let us call a shortest path between two terminals \emph{elementary} if all the internal vertices on the path are Steiner, and denote by $D_e$ all the terminal pairs that the shortest path between them is elementary. Moreover, recall that every elementary subset $S\in \T_e(G)$ is of the form $\{t_i, t_{i+1}, \ldots, t_j\}$, and denote it $S_{ij}$ and $\bar S_{ij}=S_{(j+1)(i-1)}$ for simplicity.

        By Lemma~\ref{LMA: equivalence minimum circuits and shortest paths} a shortest circuit that separates between $S^*_{ij}$ to $\bar S^*_{ij}$ in $G_1$ contains $l$ elementary cycles if and only if a shortest path between the terminals $v^{i-1,i}_{f_\infty}$ and $v^{j,j+1}_{f_\infty}$ in $G_2$ contains $l+1$ terminals (including the endpoints). Notice that Lemma~\ref{LMA: equivalence minimum circuits and shortest paths} holds also in the graphs $H_2$ and $H_1$, therefore $\T_e(G_1)=D_e(G_2)$ and $D_e(H_2)=\T_e(H_1)$. In addition, the equalities $\T_e(G)=\T_e(G_1)$ and $\T_e(H_1)=\T_e(H)$ holds by the duality between cuts and circuits, and $D_e(G_2)=D_e(H_2)$ because of the triangle inequality in the distance metric. Altogether we get that $\T_e(G)=\T_e(H)$.

        Again by the duality between cuts and circuits and by Lemma~\ref{LMA: equivalence minimum circuits and shortest paths} on the two pairs of graphs $G,G_2$ and $H_2, H_1$ we get that $\mincut_G(S_{ij})=d_{G_2}(v^{i-1,i}_{f_\infty},v^{j,j+1}_{f_\infty})$ and $\mincut_H(S_{ij})=d_{H_2}(v^{i-1,i}_{f_\infty},v^{j,j+1}_{f_\infty})$. Since $H_2$ is an $(q,s)$-DAM of $G_2$ we get that\\
        $\mincut_G(S_{ij}) \leq \mincut_H(S_{ij}) \leq q \cdot \mincut_G(S_{ij})$ and the lemma follows.
    \end{proof}

    \begin{lemma}\label{LMA: reduction H cut sparsifier of size beta}
        The size of $H$ is $O(s)$.
    \end{lemma}

    \begin{proof}
        Given that $H_2$ is an $(q,s)$-DAM, i.e. $|V(H_2)|=s$, we need to prove that $|V(H)|=O(|V(H_2)|)$. Note that by the reduction construction $|V(H)|=|F(H_1)|=|F(H_2)|+k-1$. Moreover, we can assume that $H_2$ is a simple planar graph (if it has parallel edges, we can keep the shortest one). Thus, $|E(H_2)|\leq 3|V(H_2)| +6$. Plug it in Euler's Formula to get $|F(H_2)|\leq 2|V(H_2)|+8$. Since $s\geq k$ we derive that $|V(H)|\leq 2s +8+k-1=O(s) $ and the lemma follows.
    \end{proof}

\begin{proof}[Proof of Theorem~\ref{THM: reduction from outer face cuts to DAM}]
        Given $H_2$ a $(q,s)$-DAM of $G_2$ and let $H$ be the graph that was constructed from $H_2$. By Lemma~\ref{LMA: reduction H cut sparsifier of size beta} and Lemma~\ref{LMA: reduction H is alpha cut sparsifier of G} the graph $H$ is a $(q,O(s))$-cut-sparsifier of $G$. Since $H_2$ is a minor of $G_2$, and minor is closed under planar duality, then $H$ is furthermore a minor of $G$ and the theorem follows.
    \end{proof}

\subsection{Proof of Theorem~\ref{THM: reduction from outer face DAM to outer face cuts}}
\label{app: Reduction DAM to CAM}

\paragraph{Construction of the Reduction.}
The idea is to first ``close'' the shortest paths between every two terminals into shortest cycles that separates between terminal faces, and then use the planar duality between cuts and cycles to get that every shortest cycle corresponds to a minimum terminal cut that in turn preserved by an $(q,s)$-cut-sparsifier. More formally, given a plane $k$-terminal network $G=(V,E,F,T,c)$ with all its terminals $T=\{t_1, \ldots, t_k\}$ on the outerface in a clockwise order. Firstly, construct a graph $G_1$ by adding to $G$ a new vertex $v_{f_\infty}$ and connects it to all its $k$ terminals $t_i$ using edges with 0 capacity. Note that $G_1$ has $k$ new faces $f_{1,2}, \ldots,f_{k-1,k},f_{k,1}$, where each $f_{i,i+1}$ was created by adding the two new edges $(v_{f_\infty}, t_i)$ and $(v_{f_\infty}, t_{i+1})$. These $k$ new faces will be the terminals of $G_1$. %

\begin{linenomath}
    \begin{align*}
      V(G_1) & := V\cup \{v_{f_\infty}\} \\
      E(G_1) & := E\cup \{(v_{f_\infty},t_i)\ :\ t_i\in T\} \\
      F(G_1) & := F\cup \{f_{1,2},\ldots, f_{k-1,k},f_{k,1}\}\\
      T(G_1) & := \{f_{1,2}\ ,\  \ldots\ ,\  f_{k-1,k}\ ,\  f_{k,1}\}%
    \end{align*}
\end{linenomath}
    Secondly, we denote by $G_2$ the dual graph of $G_1$, where its $k$ terminals are $T(G_2)=\{v_{i,i+1}\ :\ f_{i,i+1}\in T(G_1)\}$. Moreover, the new vertex $v_{f_\infty}$ in $G_1$ corresponds to the outerface $f_\infty$ of $G_2$, the $k$ new edges $(v_{f_\infty},t_i)$ we added to $G_1$ are the edges that lie on the outerface of $G_2$, and the vertices on the outerface of $G_2$ are the $k$ terminals $v_{i,i+1}$ in a clockwise order. See Figure \ref{FGR: Reduction Outer Mimicking To Outer DPM} from left to right for illustration, and see Appendix~\ref{app: Planar Duality} for basic notions of planar duality.

    Let $H_2$ be a $(q,s)$-cut-sparsifier and a minor of $G_2$. Since $H_2$ is a minor of $G_2$, then both are plane graphs with all their terminals on the outerface in the same clockwise order, and there is an edge with capacity 0 on the outerface that connects between every two adjacent terminals. Hence, we can use $H_2$ and the same reduction above (but in opposite order of operations) in order to construct an $(q,O(s))$-DAM $H$ of $G$ as follows. Firstly, let $H_1$ be the dual graph of $H_2$, where every minimum terminal cut in $H_2$ is equivalent to a shortest cycle that separates terminal faces. Notice that again each terminal face $f_{i,i+1}$ in $H_1$ contains the two edges $(v_{f_\infty}, t_i)$ and $(v_{f_\infty}, t_{i+1})$ with capacity 0 on their boundary. Secondly, we ``open'' each shortest cycle in $H'^*$ into a shortest path between terminals by removing the vertex $v_{f_\infty}$ and all its incidence edges, and denote this new graph by $H$. The terminals of $H$ are all the vertices $v\in V(H_1)$ such that $(v_{f_\infty},v)$ is an edge in $H_1$, which are equal to the original terminals of $G$. See Figure \ref{FGR: Reduction Outer DPM to Outer Cuts} from right to left for illustration.

\begin{linenomath}
        \begin{align*}
          V(H) & := V(H_1)\setminus \{v_{f_\infty}\} \\
          E(H) & := E(H_1)\setminus \{(v_{f_\infty},t_i)\ :\ t_i\in T\} \\
          F(H) & :=F(H_1)\setminus \{f_{1,2},\ldots, f_{k-1,k},f_{k,1}\}\\
          T(H) & := T(G)
        \end{align*}
\end{linenomath}

     \begin{figure}
        \centering
        \includegraphics[angle=0,width=0.9\textwidth]{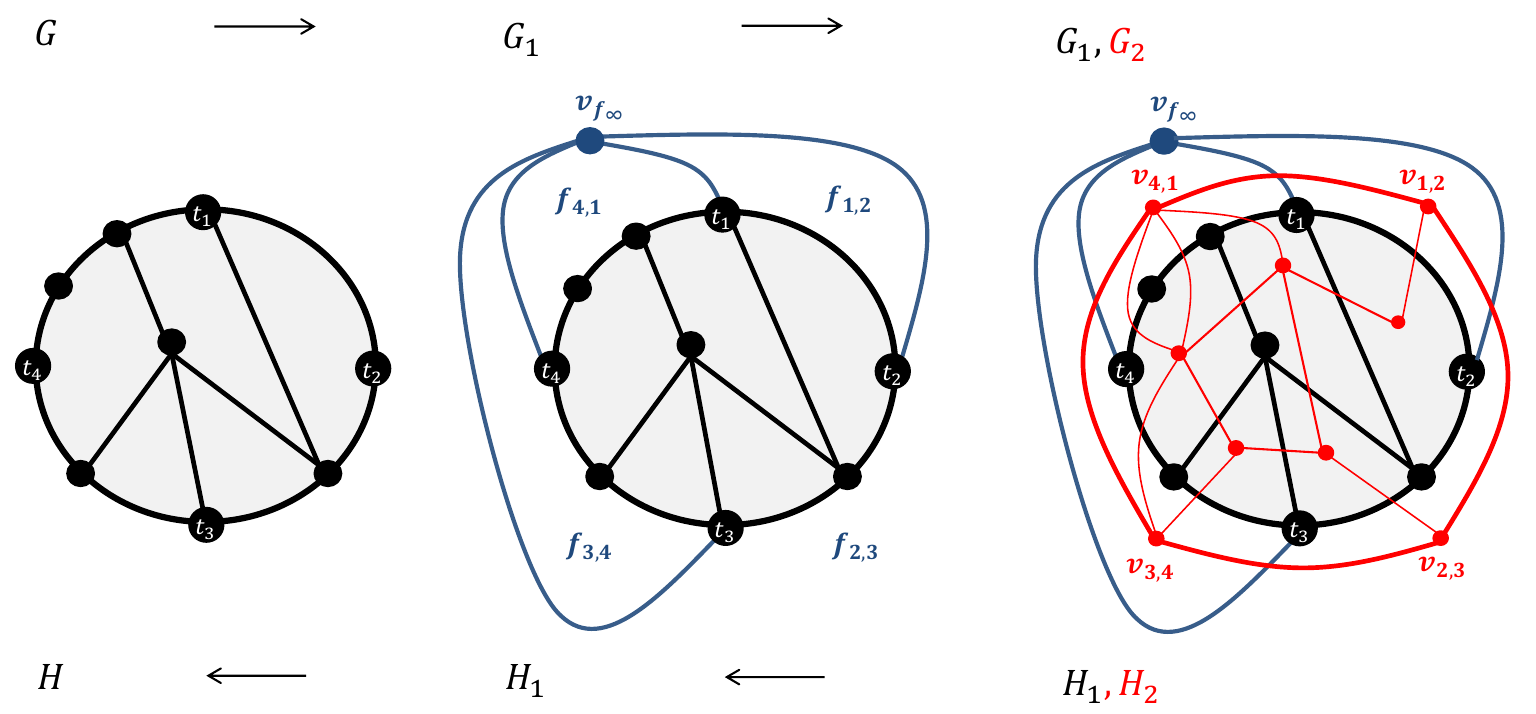}
        \caption{The first graph (in black) is the original graph $G$. The second is the graph $G_1$, where the additional vertex and edges and terminal faces colored in blue. And the third graph $G_2$ and its terminals are colored in red. The bold red edges are the dual of the blue edges, and both of them have capacity 0.}
        \label{FGR: Reduction Outer DPM to Outer Cuts}
        \mbox{}\hrule
     \end{figure}

\paragraph{Analysis of the Reduction.}

\begin{lemma}\label{LMA: reduction H DAM of size beta}
The size of $H$ is $O(s)$.
\end{lemma}

\begin{proof}
Given that $H_2$ is an $(q,s)$-cut-sparsifier, i.e. $|V(H_2)|=s$, we will prove that $|V(H)|=O(|V(H_2)|)$. We can assume that $H_2$ is a simple planar graph (if not, we can replace all the parallel edges between every two vertices by one edge where its capacity is the sum over all the capacities of these parallel edges), thus $|E(H_2)|\geq \frac{3}{2}|F(H_2)|$. Plug it in Euler's Formula to get $|F(H_2)|\leq 2|V(H_2)|-4=2s-4$. By the reduction construction $|V(H)|+1=|V(H_1)|=|F(H_2)|=O(s)$, and the lemma follows.
\end{proof}

\begin{lemma}\label{LMA: H is a minor of G}
The graph $H$ is a minor of $G$.
\end{lemma}

\begin{proof}
Given that $H_2$ is a minor of $G_2$, and that minor is close under deletion and contraction of edges we get that $H_1$ is a minor of $G_1$. Now by deleting the same vertex $v_{f_\infty}$ together with all its incidence edges from both $G_1$ and $H_1$, we get the graphs $G$ and $H$ correspondingly. Therefore $H$ is a minor of $G$, and the lemma follows.
\end{proof}

\begin{lemma}\label{LMA: reduction H is alpha DAM of G}
The graph $H$ preserve all the distances between every two terminals by factor $q$, i.e. $d_G(t_i,t_j) \leq d_H(t_i,t_j) \leq q \cdot d_G(t_i,t_j)$ for every $t_i,t_j\in T$.
\end{lemma}

\begin{proof}
Notice that connecting all the terminals to a new vertex using edges with capacity 0 is equivalent to uniting all the terminals into one vertex, and also splitting the vertex $v_{f_\infty}$ to $k$ new terminals is equivalent to disconnecting all the terminals by deleting that vertex. Thus our reduction is equivalent to the reduction of Theorem \ref{THM: reduction from outer face cuts to DAM}. In particular, Lemma \ref{LMA: equivalence minimum circuits and shortest paths} holds on the graphs $G,G_1$ and on the graphs $H, H_1$ correspondingly, i.e. every shortest path between two terminals $t_i$ and $t_j$ in $G$ (or $H$) corresponds to a minimum circuit in $G_1$ (or $H_1$) that separates between the terminal faces $\{f_{i,i+1}, \ldots, f_{j-1,j}\}$ to $\{f_{j,j+1}, \ldots, f_{i-1,i}\}$ and vise versa.

Let $S_{(i,i+1),(j-1,j)}=\{v_{i,i+1},\dots, v_{j-1,j}\}$ be a set of terminals in $G_2$, where every terminal $v_{l,l+1}$ corresponds to the terminal face $f_{l,l+1}$ in $G_1$. By the duality between cuts and circuits we get that $d_{G}(t_i,t_j)=\mincut_{G_2}(S_{(i,i+1),(j-1,j)})$ and $d_{H}(t_i,t_j)=\mincut_{H_2}(S_{(i,i+1),(j-1,j)})$. Since $H_2$ is an $(q,s)$-cut-sparsifier of $G_2$ we derive the inequalities
$d_G(t_i,t_j) \leq d_H(t_i,t_j) \leq q \cdot d_G(t_i,t_j)$
and the lemma follows.
\end{proof}

\begin{proof}[Proof of Theorem \ref{THM: reduction from outer face DAM to outer face cuts}]
By Lemma~\ref{LMA: reduction H DAM of size beta}, Lemma \ref{LMA: H is a minor of G} and Lemma \ref{LMA: reduction H is alpha DAM of G} the $k$-terminal network $H$ is an $(q,O(s))$-DAM of $G$ and the theorem follows.
\end{proof}

\subsection{Duality Applications}\label{SEC: Reduction Application}

By Theorem~\ref{THM: reduction from outer face cuts to DAM}
and Theorem~\ref{THM: reduction from outer face DAM to outer face cuts},
every $k$-terminal network $G$ with $\gamma=1$ admits a $(q,s)$-DAM
if and only if it admits a minor $(q,O(s))$-cut-sparsifier.
Hence, every new upper or lower bound results, especially for $q>1$, on DAM also holds for the minor cut-sparsifier problem and vise versa.
For example, the upper bound of $(1+\epsilon, (k/ \epsilon)^2)$-DAM for planar networks \cite{CGH16} yields the following new theorem.

\begin{theorem} \label{THM: planar CAM epsilon-approx}
Every planar network $G$ with $\gamma=1$ admits a minor $(1 + \epsilon, \tilde{O}((k/\epsilon)^2)$-cut-sparsifier for every $\epsilon>0$.
\end{theorem}
As already mentioned, by recent independent work \cite{GHP17}
these networks also admit a $(1,O(k^2))$-sparsifier
that is planar but \emph{not} a minor of $G$.

In addition, we can apply known upper and lower bounds for $(1,s)$-DAM
to the minor mimicking network problem (i.e., a cut-sparsifier of quality 1).
In particular, the known $(1,k^4)$-DAM \cite{KNZ14} yields
an alternative proof for Corollary~\ref{THM: mimick network size k^4},
and the known lower bound of $(1,\Omega(k^2))$-DAM
(which is shown on grid graphs) \cite{KNZ14}
yields an alternative proof for a lower bound shown in \cite{KR13SODA}.

\subsubsection*{Acknowledgments}
We thank anonymous referees for useful suggestions that improved the presentation.

{%
\bibliographystyle{alphaurlinit}
\bibliography{inbal}
}

\appendix

\section{Planar Duality}\label{app: Planar Duality}

Using planar duality we bound the size of mimicking networks for planar graphs (Theorem~\ref{THM: planar mimick network size alpha|U|^2}), and we further use it to prove the duality between cuts in distances (Theorem~\ref{THM: reduction from outer face cuts to DAM} and Theorem~\ref{THM: reduction from outer face DAM to outer face cuts}) Recall that every planar graph $G$ has a dual graph $G^*$,
    whose vertices correspond to the faces of $G$,
    and whose faces correspond to the vertices of $G$,
    i.e., $V(G^*)=\{v^*_f: f\in F(G)\}$ and $F(G^*)=\{f^*_v : v\in V(G)\}$.
    Thus the terminals $T=\{t_1,\ldots,t_k\}$ of $G$ corresponds to the terminal faces $T(G_1)=\{f_{t_{1}},\ldots,f_{t_{k}}\}$ in $G^*$, which for the sake of simplicity we may refer them as terminals as well. Every edge $e=(v,u)\in E(G)$ with capacity $c(e)$ that lies on the boundary of two faces $f_1,f_2\in F(G)$ has a dual edge $e^*=(v^*_{f_1},v^*_{f_2})\in E(G^*)$ with
    the same capacity $c(e^*)=c(e)$ that lies on the boundary of the faces $f^*_v$ and $f^*_u$.
    For every subset of edges $M\subset E(G)$, let $M^*:=\{e^*: e\in M\}$ denote the subset
    of the corresponding dual edges in $G^*$.

    The following theorem describes the duality between two different kinds
    of edge sets -- minimum cuts and minimum circuits -- in a plane multi-graph.
    It is a straightforward generalization of the case of $st$-cuts
    (whose dual are cycles) to three or more terminals.

    A \emph{circuit} is a collection of cycles (not necessarily disjoint) $\mathcal{C}=\{C_1,\ldots,C_l\}$. Let $\mathcal{E(C)}=\cup_{i=1}^l C_i$ be the set of edges that participate in one or more cycles in the collection
    (note it is not a multiset, so we discard multiplicities).
    The capacity of a circuit $\mathcal{C}$ is defined as $\sum_{e\in \mathcal{E(C)}}c(e)$.

    \begin{theorem}[Duality between cutsets and circuits]
    \label{THM: duality cuts and circuits}
    Let $G$ be a connected plane multi-graph, let $G^*$ be its dual graph,
    and fix a subset of the vertices $W\subseteq V(G)$.
    Then, $M\subset E(G)$ is a cutset in $G$ that has minimum capacity among those
    separating $W$ from $V(G)\setminus W$
    if and only if the dual set of edges $M^*\subseteq E(G^*)$
    is actually $\mathcal E(\mathcal C)$ for a circuit $\mathcal C$ in $G^*$
    that has minimum capacity among those separating the corresponding faces
    $\{f^*_v: v\in W\}$ from $\{f^*_v: v\in V(G) \setminus W\}$.
    \end{theorem}

    \begin{lemma}[The dual of a connected component]
    \label{LMA: connected components vs regions}
    Let $G$ be a connected plane multi-graph, let $G^*$ be its dual,
    and fix a subset of edges $M\subset E(G)$.
    Then $W\subseteq V$ is a connected component in $G\setminus M$ if and only if
    its dual set of faces $\{f^*_v: v\in W\}$ is a face of $G^*[M^*]$.
    \end{lemma}

    Fix $S\subset T$. We call $E^*_S$ elementary circuit if $E_S$ is an elementary cutset in $G$. Note that by Lemma \ref{LMA: connected components vs regions} $E^*_S$ is an elementary circuit if and only if the graph $G^* \setminus E^*_S$ has exactly two faces. Thus $\T_e(G)=\T_e(G^*)$, and the circuit $E^*_S$ has exactly one minimum cycle in $G^*$. For the sake of simplicity we later on use the term cycle instead of circuit when we refer to elementary minimum circuit.

\end{document}